\documentclass[aps,reprint,twocolumn,superscriptaddress,pra,nofootinbib]{revtex4-2}
\usepackage[caption=false]{subfig}
\usepackage{amsmath,amsfonts,amsthm,graphicx,mathtools,epstopdf,array,etoolbox,makecell}
\usepackage[T1]{fontenc}
\usepackage{booktabs}

\usepackage{silence}
\usepackage{braket}
\newcommand{\ketbra}[2]{\ket{#1} \! \bra{#2}}
\newcommand{\Tr}{\operatorname{Tr}}

\usepackage{hyperref} 
\usepackage{xurl} 

\hypersetup{
breaklinks=true,
colorlinks = true,
citecolor= blue,
urlcolor= blue,
linkcolor = blue
}

\usepackage{algorithm} 
\usepackage[noend]{algpseudocode}
\floatname{algorithm}{Protocol} 
\algrenewcommand\alglinenumber[1]{\normalsize #1.} 

\newcounter{algsubstate}

\newenvironment{algsubstates}
  {\setcounter{algsubstate}{0}%
   \renewcommand{\State}{%
     \refstepcounter{algsubstate}%
     \Statex {\normalsize\arabic{ALG@line}.\arabic{algsubstate}.}\kern5pt}
     }
  {}

\makeatletter
\renewcommand\onecolumngrid{%
  \do@columngrid{one}{\@ne}%
  \def\set@footnotewidth{\onecolumngrid}%
  \def\footnoterule{\kern-6pt\hrule width 1.5in\kern6pt}%
}
\renewcommand\twocolumngrid{%
  \def\footnoterule{%
    \dimen@\skip\footins
    \divide\dimen@\thr@@
    \kern-\dimen@
    \hrule width .5in
    \kern\dimen@
  }%
  \do@columngrid{mlt}{\tw@}%
}
\makeatother

\newcommand{\Renyi}{R\'{e}nyi}
\newcommand{\mbf}[1]{\mathbf{#1}} 
\newcommand{\bsym}[1]{\boldsymbol{#1}} 

\newcommand{\id}{\mathbb{I}} 
\newcommand{\qlowperp}{q_\mathrm{low}(\perp)}
\newcommand{\num}{\operatorname{num}}
\newcommand{\EUR}{F}
\newcommand{\sixstate}[1]{G_{#1}}
\newcommand{\epsperp}{\varepsilon^\perp}

\newtheorem{remark}{Remark}

\newtheorem{lemma}{Lemma}
\newtheorem{corollary}{Corollary}

\theoremstyle{definition} 

\begin{document}

\title{Simple QKD keyrate computations from tangent-line bounds}

\author{Goh Jun Hui}
\email{junhui.goh24@sps.nus.edu.sg}
\affiliation{Department of Physics, National University of Singapore}
\affiliation{Special Programme in Science, National University of Singapore}

\author{Ernest Y.-Z. Tan}
\email{phytyze@nus.edu.sg}
\affiliation{Department of Physics, National University of Singapore}

\begin{abstract}
Recent work in quantum key distribution (QKD) has yielded finite-size keyrates based on convex optimisation problems, using the framework of entropy accumulation. Multiple solvers have been developed to address these optimisations in recent years, but have been based on a somewhat elaborate framework using sophisticated algorithms, in order to accommodate a broad range of QKD protocols. In this work, we note that for basic QKD protocols such as qubit BB84, the keyrates can be easily computed using off-the-shelf solvers instead, by obtaining suitable tangent lines to lower bound terms in the objective function.
We compute the resulting finite-size keyrates and compare them to previous work. 
Moreover, we extend this approach to the six-state protocol, by deriving closed-form expressions for the single-round R\'{e}nyi entropies in the protocol. This should facilitate subsequent study of the six-state protocol via entropy accumulation.
\end{abstract}

\maketitle

\section{Introduction}

Quantum Key Distribution (QKD) allows two parties to establish a secret key with security based on quantum mechanics. However, a thorough security analysis of a QKD protocol should account for the fact that it only runs for a finite number of rounds $n$, and it should ensure security even if the attacker's strategy is not independent and identically distributed (IID) across the rounds. These ``finite-size effects'' can have a significant impact on the keyrate, which should be accounted for in the security proof. 

Recent work based on entropy accumulation~\cite{arqand_generalized_2025} has yielded finite-size keyrates for a broad range of QKD protocols. In this framework, the evaluation of these keyrates is reduced to a particular convex optimisation problem, involving the {\Renyi} entropy~\cite{tomamichel_quantum_2016} accumulated in a single round. Multiple solvers and software packages have since been developed for the task of addressing this optimisation~\cite{he_qics_2024,he_operator_2026,kamin_renyi_2025,navarro_finite-size_2025}. Moreover, by making use of duality theorems from convex optimisation~\cite{boyd_convex_2004}, they can provide certified lower bounds on the optimal value, thereby ensuring rigorous lower bounds on the achievable keyrate (i.e.~ensuring we have not over-estimated the secure keyrate).
However, these solvers are generally rather sophisticated, as they are designed to apply generically to many different QKD protocols. 

In this work, we demonstrate that for simple qubit protocols such as BB84 and six-state, these convex optimisations can be more easily addressed using off-the-shelf solvers for convex optimisation.
For the BB84 protocol, our approach is to first note that the keyrate computation written in~\cite{arqand_generalized_2025} for that specific protocol is fairly simple, though it involves one ``nonstandard'' convex function in the objective function.
In that work, that computation was simply evaluated using heuristic optimisation methods, which were straightforward but did not yield rigorously certified lower bounds. 
We address this issue in this work, while maintaining speed and simplicity, by implementing it in an off-the-shelf convex solver. To do so, we note that the ``nonstandard'' term in the objective function is a function of just a single real variable.
It is therefore straightforward to find its tangent lines, and taking the pointwise maximum of these tangent lines evaluated at any finite set of points yields a valid lower bound on the original convex function. Such a pointwise maximum is straightforwardly implemented in off-the-shelf solvers, allowing us to obtain certified bounds on the keyrate.

We compare our results against those computed in~\cite{arqand_generalized_2025} using heuristic optimisation methods, and find that they are very similar, so there was only a minimal loss in keyrate incurred by finding these certified values. Next, we additionally consider a slightly improved formula for the keyrate as compared to the one computed in that work, which is based on a slightly different variant of {\Renyi} entropy~\cite{tomamichel_quantum_2016} from the latter.
We find it indeed yields slightly better keyrates in practice at small $n$, similar to previous observations in e.g.~Ref.~\cite{ANB24}. Furthermore, we plot the optimal values we found for the test-round probability and {\Renyi} parameter as a function of $n$, obtaining an empirical estimate of their scaling behaviours.

For the six-state protocol, our contribution is to first obtain closed-form expressions for the {\Renyi} entropy (for arbitrary $\alpha\in(1,2)$) in single rounds of the protocol, a task which has not previously been performed. Similar to the BB84 protocol, we find that it is a convex function of a single real variable, which we can then differentiate to obtain its tangent lines. We then implement the resulting keyrate computation in a convex solver in similar fashion.

This paper is organised as follows. In Sec.~\ref{sec:BB84analysis}, we present the framework we use to analyse the BB84 protocol (in entanglement-based form). In Sec.~\ref{sec:BB84results}, we present the resulting keyrates we obtained, together with empirical observations of the scaling of the optimal parameter choices. In Sec.~\ref{sec:sixstateanalysis}, we present our analysis of the six-state protocol, with the derivation of the closed-form expression for {\Renyi} entropy being deferred to Appendix~\ref{app:sixstate}. In Sec.~\ref{sec:sixstateresults}, we present the resulting keyrates and empirically optimal parameters. Finally, in Sec.~\ref{sec:conclusion} we provide concluding remarks and discussion.

\section{Analysis of BB84 protocol}
\label{sec:BB84analysis}

\subsection{Protocol Description}

\newcommand{\lEC}{\lambda_\mathrm{EC}}
\newcommand{\lkey}{\ell_\mathrm{key}}
\newcommand{\CP}{\widehat{C}} 
\newcommand{\cP}{\hat{c}} 
\newcommand{\term}[1]{\textup{\textbf{#1}}}

Following~\cite{arqand_generalized_2025}, we consider an entanglement-based BB84 protocol, as outlined in Protocol~\ref{prot:EB_BB84} below. This is a ``fixed-length'' protocol which simply makes a binary accept/abort decision and outputs a key of a predetermined fixed length $\lkey$ whenever it accepts. (This is in contrast to ``variable-length'' protocols where the length of the output key can vary depending on the observations in the protocol.) We use $\sigma_X, \sigma_Y, \sigma_Z$ to denote the Pauli operators, and describe qubit measurements in those terms.

\begin{algorithm}[H]
\caption{EB-BB84 protocol outline. Reproduced with permission from~\cite{arqand_generalized_2025}, with minor modifications.}
\label{prot:EB_BB84}
\begin{algorithmic}[1]
\State For each round $j \in \{1,2,\dots,n\}$, perform the following steps:
\begin{algsubstates}
\State Alice and Bob each receive and briefly store a qubit register. Then via public communication, with probability $\gamma \in (0,1)$ (independently in each round) they jointly declare the round is a test round, and otherwise it is a generation round. Alice also independently generates a uniformly random ``symmetrization bit'' $F_j \in \{0,1\}$ and publicly announces it.
\State If it is a generation round, they both measure $\sigma_Z$, 
XOR their outcomes with $F_j$,
and store the resulting values in registers $S_j$ and $\widetilde{S}_j$ for Alice and Bob respectively. Also, they jointly set a classical register $\CP_j=\perp$. 
\State If it is a test round, they both measure $\sigma_X$, 
XOR their outcomes with $F_j$ 
and publicly announce the resulting values, then jointly set a classical register $\CP_j=0$ if their outcomes matched and $\CP_j=1$ otherwise. Also, Alice sets $S_j=0$ and Bob sets $\widetilde{S}_j = \texttt{test}$.
\end{algsubstates}
\State Perform an \term{acceptance test}, in which Alice and Bob abort if and only if the frequency distribution on ${\CP_1^n}$ lies outside some predetermined convex set $S_\Omega$.
\State Perform one-way \term{error correction} in which Alice sends Bob a bitstring of length $\lEC$, which he uses together with $\widetilde{S}_1^n$ to produce a guess $\mbf{S}^{\mathrm{guess}}$ for Alice's string $S_1^n$. Then perform \term{error verification}, in which Alice sends Bob a 2-universal hash of $S_1^n$ (together with the choice of hash function), who compares it with the hash of $\mbf{S}^{\mathrm{guess}}$ and aborts if they do not match.
\State If neither of the above steps aborted, produce final keys of length $\lkey$ by performing \term{privacy amplification} on $S_1^n$ and $\mbf{S}^{\mathrm{guess}}$.
\end{algorithmic}
\end{algorithm}

We do not discuss the error correction and privacy amplification steps in detail here; refer to~\cite{tomamichel_largely_2017} for an in-depth explanation. However, we elaborate further on the acceptance test as we will need it for some subsequent discussions. We first highlight that each of the $\CP_j$ registers has the same alphabet $\{0,1,\perp\}$; for subsequent notational convenience, let us introduce a single (classical) register $\CP$ with that alphabet. By ``frequency distribution on $\CP_1^n$'', we mean the vector
\begin{align}
\label{eq:freqCn}
\frac{1}{n} \left(\num_{\CP_1^n}(0), \num_{\CP_1^n}(1), \num_{\CP_1^n}(\perp) \right) ,
\end{align}
where $\num_{\CP_1^n}(j)$ denotes the number of occurrences of the symbol $j$ in the string $\CP_1^n$. Notice that~\eqref{eq:freqCn} is a valid probability distribution on the classical register $\CP$. It thus makes sense for us to speak of some convex set $S_\Omega$ of probability distributions on that register, which is fixed before the protocol begins (i.e.~it is not allowed to change after observing any measurement outcomes in the protocol) and specifies all the frequency distributions on $\CP_1^n$ to be accepted in the acceptance-test step.\footnote{While the set of possible frequency distributions on $\CP_1^n$ is a finite ``discrete'' set, it is still valid to speak of whether any given frequency distribution from it lies in $S_\Omega$ (a convex set of probability distributions on $\CP$), since it will be a valid probability distribution as discussed above.}
We shall discuss in later sections the exact choice of this set $S_\Omega$.

The security of a QKD protocol is quantified by a concept known as $\varepsilon_{\mathrm{secure}}$-security; we do not discuss it further here but defer the details to Ref.~\cite{PR22} (it is referred to as ``soundness'' in that work).
Given a desired security level $\varepsilon_{\mathrm{secure}}$,
the following formula was derived\footnote{Here we have corrected a typo in the formula from that work: the infimum over $\boldsymbol{\nu}$ should have been taken over the set $\Theta_{\CP}$ we describe below, \emph{not} the entire set of probability distributions on $\CP$ as claimed in that work.}
in~\cite[Eq.~(145)]{arqand_generalized_2025} as a (lower) bound on the achievable keyrate of Protocol~\ref{prot:EB_BB84}, for any $\alpha\in(1,2)$:
\begin{equation}
\label{eq:keyrate}
\begin{aligned}
\lkey &= \inf_{\substack{
\boldsymbol{\nu} \in \Theta_{\CP}\\
\mathbf{q} \in S_{\Omega}
}} 
\Bigg[
 q(\perp) \EUR_\beta \left(\frac{\nu(1)}{\gamma}\right)
+ 
K_\alpha
D\!\left(\mbf{q} \parallel \boldsymbol{\nu}\right)
\Bigg] n
\\
&\quad
- \lambda_{\mathrm{EC}}
- \log\!\left(\frac{1}{\varepsilon_{\mathrm{EV}}}\right)
- \frac{\alpha}{\alpha-1}
\log\!\left(\frac{1}{\varepsilon_{\mathrm{PA}}}\right)
+ 2, \\
&\text{for } \EUR_\beta(Q) \coloneqq
1 - \frac{1}{1-\beta}
\log\Big(
\left(1 - Q\right)^{\beta}
+ Q^{\beta}
\Big),
\end{aligned}
\end{equation}
where the various terms are defined as follows. (Here we only state the final expressions for these terms; refer to~\cite{arqand_generalized_2025} for an explanation of the reasoning behind these values.) 

First, $\Theta_{\CP}$ denotes the set of probability distributions on $\CP$ with the property that the probability of $\perp$ is exactly $1-\gamma$. Recalling that the alphabet of that register is $\{0,1,\perp\}$, we have used the following notation for the components of the variable $\boldsymbol{\nu} \in \Theta_{\CP}$ (and similarly for $\mbf{q} \in S_\Omega$):
\begin{align}
\label{eq:nu_and_q}
\begin{gathered}
\boldsymbol{\nu} = (\nu(0), \nu(1), \nu(\perp)) \quad \text{(where $\nu(\perp) = 1-\gamma$)},\\
\mbf{q} = (q(0), q(1), q(\perp)).
\end{gathered}
\end{align}
Moreover, $D\!\left(\mbf{q} \parallel \boldsymbol{\nu}\right)$ denotes the Kullback-Leibler (KL) divergence~\cite{cover_elements_2005} between those two distributions.

Next, $\lambda_{\mathrm{EC}}$ is the length of a string used for error correction in the protocol, and has the value
\begin{equation}
    \lambda_{\mathrm{EC}} = \xi(1-\gamma)h_{\mathrm{bin}}(Q_{\mathrm{thresh}})n,
\label{relation 5}
\end{equation}
where $h_{\mathrm{bin}}$ is the binary entropy function, and $\xi$ is a parameter that describes the ``efficiency'' of the error-correction procedure (see e.g.~\cite{TMP+17}); we provide specific values in Sec.~\ref{subsec:protparams}.
Furthermore, the values of $\varepsilon_{\mathrm{EV}}$ and $\varepsilon_{\mathrm{PA}}$ are chosen as follows:
\begin{equation}
    \varepsilon_{\mathrm{EV}} = \frac{\alpha - 1}{2\alpha - 1}\varepsilon_{\mathrm{secure}},
    \quad
    \varepsilon_{\mathrm{PA}} = \frac{\alpha}{2\alpha - 1}\varepsilon_{\mathrm{secure}},
\label{relation 3}
\end{equation}
as they yield the best keyrate for any given values of $\alpha,\varepsilon_{\mathrm{secure}}$~\cite{arqand_generalized_2025}.

Finally, the parameters $\beta$ and $K_\alpha$ were taken in Ref.~\cite{arqand_generalized_2025} to be
\begin{align}
\label{eq:original parametrization}
\beta = \frac{1}{\hat{\alpha}}, \quad K_\alpha = \frac{1}{\hat{\alpha}-1}, \qquad \text{for } \hat{\alpha} = \frac{1}{2-\alpha},
\end{align}
however, we highlight here that slightly better keyrates can be obtained by instead setting them to
\begin{align}
\label{eq:new parametrization}
\beta = \frac{\alpha}{2\alpha - 1}, \quad K_\alpha = \frac{\alpha}{\alpha-1},
\end{align}
with the validity of this choice being justified by using Theorem~VI.1 in Ref.~\cite{arqand_generalized_2025} in place of Theorem~V.1 of that work.\footnote{The improvement arises from using a different version of {\Renyi} entropy (see Ref.~\cite{arqand_generalized_2025} for further discussion), yielding these improved parameter dependencies. The computations in Ref.~\cite{arqand_generalized_2025} were only based on Theorem~V.1 of that work because the focus in that work was the performance of that theorem specifically, rather than fully optimizing the keyrate.}

As $n\to\infty$, one can take $\alpha\to 1$ at a suitable rate~\cite{arqand_generalized_2025} to obtain the standard formula for the asymptotic keyrate of the BB84 protocol (as presented in e.g.~\cite{SBC+09}):
\begin{align}
r_{\mathrm{asymp}} = 1 - 2 h_{\mathrm{bin}}(Q_{\mathrm{thresh}}).
\end{align}

For the purpose of numerical computations, we note that the minimisation in the keyrate formula~\eqref{eq:keyrate} can be conveniently rewritten into the following equivalent form, keeping in mind the structure of $\bsym{\nu}$ from~\eqref{eq:nu_and_q}:
\begin{equation}
\label{eq:infoverQ}
\begin{aligned}
\inf_{\substack{
Q\in[0,1]\\
\mathbf{q} \in S_{\Omega}
}} 
\Bigg[
 q(\perp) \EUR_\beta \left(Q\right)
+ 
K_\alpha
D\!\left(\mbf{q} \parallel \boldsymbol{\nu}\right)
\Bigg] \\
\text{where } \boldsymbol{\nu} = (\gamma(1-Q), \gamma Q, 1-\gamma).
\end{aligned}
\end{equation}
We now discuss how to evaluate this optimisation.

\subsection{Convex Optimisation}

To compute a certified value for the achievable keyrate, ideally we would like to make use of the framework of convex optimisation, which in principle allows one to certify lower bounds on the computed value through the concept of \term{dual values}~\cite{boyd_convex_2004}. However, in order to achieve this, we need to have a convex objective function to evaluate.

We first note that the $K_\alpha D\!\left(\mbf{q} \parallel \boldsymbol{\nu}\right)$ term in~\eqref{eq:infoverQ} is convex with respect to the optimisation variables $(Q,\mbf{q})$, because KL divergence is a convex function~\cite{cover_elements_2005} and $\bsym{\nu}$ is an affine function of $Q$. Next, we observe that the expression
\begin{equation}
\label{eq:renyi entropy}
\frac{1}{1-\beta}
\log\Big(
\left(1 - Q\right)^{\beta}
+ Q^{\beta}
\Big)
\end{equation}
is simply the {\Renyi} entropy (see~\cite{tomamichel_quantum_2016} for an extensive discussion) of the probability distribution $(1-Q,Q)$, with $\beta<1$. Since {\Renyi} entropy is concave for such values of $\beta$~\cite{tomamichel_quantum_2016}, this implies that~\eqref{eq:renyi entropy} is concave (with respect to $Q$), and it follows that the $\EUR_\beta(Q)$ term in~\eqref{eq:infoverQ} is convex.

However, what appears in the minimisation in~\eqref{eq:infoverQ} is not just $\EUR_\beta\left(Q\right)$, but rather a product of the form $q(\perp) \EUR_\beta\left(Q\right)$. As we show in Appendix~\ref{app:non-convex}, this product is \emph{not} (jointly) convex in $q(\perp)$ and $Q$, thus posing an obstruction for applying convex optimisation methods to this minimisation.\footnote{In principle it might be possible that the \emph{sum} of this term and the KL divergence term $K_\alpha D\!\left(\mbf{q} \parallel \boldsymbol{\nu}\right)$ is still convex, but we leave this possibility to be resolved in future work.} 

To overcome this issue, we observe that the minimisation over $\mbf{q}$ is constrained to lie within $S_\Omega$, i.e.~the set describing the frequency distributions accepted in the protocol. 
If this set $S_\Omega$ is such that there is some constant $\qlowperp$ that lower-bounds $q(\perp)$, then we could simply replace the latter by the former in the optimisation, thereby obtaining a convex expression $\qlowperp \EUR_\beta\left(Q\right)$ in place of $q(\perp) \EUR_\beta\left(Q\right)$. Hence we shall now discuss this aspect.

\subsubsection{Designing the acceptance test}

We now return to the question of choosing the set $S_\Omega$. In Ref.~\cite{arqand_generalized_2025} (following Ref.~\cite{tomamichel_largely_2017}), it was defined as follows, writing $\mathbb{P}_{\CP}$ to denote the set of probability distributions on $\CP$:
\begin{equation}
S^\mathrm{previous}_\Omega = \{ \mathbf{q} \in \mathbb{P}_{\CP} | q(1) \leq \gamma Q_\mathrm{thresh} \},
\end{equation}
for a predetermined ``threshold'' value $Q_\mathrm{thresh}$. Qualitatively, this says that the protocol aborts when the frequency of rounds that are test rounds with ``errors'' (non-matching measurement outcomes) exceeds $\gamma Q_\mathrm{thresh}$.

However, this set $S^\mathrm{previous}_\Omega$ does not lower bound the component $q(\perp)$ in the manner we would have wanted from the preceding discussion. To allow the construction of such a lower bound, we instead use an $S_\Omega$ that has an additional abort condition, qualitatively corresponding to a low rate of generation rounds.
Note that this means that we are technically considering a slightly different protocol from Ref.~\cite{arqand_generalized_2025}, in that our protocol is slightly more likely to abort (since our acceptance criterion is stricter). However, via the subsequent analysis, we shall also construct this abort condition such that it only occurs with small probability, so we believe it is still reasonable to compare the resulting keyrates against those obtained in that work.

To begin, we observe that if we denote the number of generation rounds as $X$, it can be described by a binomial distribution~\cite{blitzstein_introduction_2014},
\begin{equation*}
    X \sim B(n, 1-\gamma).
\end{equation*}
The expectation value of $X$ is $n(1-\gamma)$. We shall now introduce a parameter $\delta$ such that the protocol aborts if $\frac{X}{n} < 1-\gamma-\delta$, which is to say we define the set $S_\Omega$ as 
\begin{equation}
\label{eq:new acceptance set}
\begin{aligned}
S_\Omega = 
& \{ \mathbf{q} \in \mathbb{P}_{\CP} | q
(1) \leq \gamma Q_\mathrm{thresh} \\
& \text{and } q
(\perp) \geq 1-\gamma-\delta\}.
\end{aligned}
\end{equation}

To ensure that the protocol is unlikely to abort from this additional condition, we shall choose $\delta$ as follows. First, pick some desired small probability $\epsperp$ (we shall use the choice $\epsperp=0.01$ in the later computations). We then define $\delta$ to be the value satisfying the following equation:
\begin{equation}
    \epsperp = \Pr_{X \sim B(n, 1-\gamma)} \left( \frac{X}{n} < 1 - \gamma -\delta \right),
\end{equation}
noting that the required $\delta$ value can be solved for by using standard computational-software functions for binomial distributions. Recalling that the number of generation rounds indeed always\footnote{For readers already familiar with QKD security definitions, we highlight that the \term{completeness} property~\cite{PR22} is, by definition, only determined by the abort probability of the honest behaviour (i.e.~without Eve), usually assumed to be IID. Our claim regarding $\epsperp$ here is slightly stronger, in that it applies even in the presence of Eve --- this is because the test/generation decision is performed using trusted randomness by Alice and Bob, and is thus always IID regardless of Eve.} follows the binomial distribution here, we see that this ensures the probability of violating the $q(\perp) \geq 1-\gamma-\delta$ condition in~\eqref{eq:new acceptance set} is at most $\epsperp$.\footnote{If desired, we can say in turn that the union bound implies that the probability of this protocol aborting is at most $\epsperp$ larger than the protocol in Ref.~\cite{arqand_generalized_2025}.} Thus we believe it is reasonable to view the introduction of this additional condition as having only a small effect on the implementation of the protocol.

Using this set-up, the lower bound $\qlowperp$ can be formulated as 
\begin{equation} \label{eq:q lower bound}
    \qlowperp = 1 - \gamma - \delta,
\end{equation}
i.e.~the value of $q(\perp)$ 
for any $\mbf{q}\in S_\Omega$ will be lower bounded by $\qlowperp$. The expression 
\begin{equation}
\label{eq:convex lower}
\begin{aligned}
&
\qlowperp
\EUR_\beta(Q) + K_\alpha
D\!\left(\mbf{q} \parallel \boldsymbol{\nu}\right)
\end{aligned}
\end{equation}
is therefore a convex lower bound on the objective function of the original minimisation~\eqref{eq:infoverQ}. 

\subsubsection{Tangent-line lower bounds}

To automatically exploit the existing algorithms and convergence guarantees for convex optimisation, we wish to implement the above optimisation in a modelling language such as CVXPY~\cite{diamond2016cvxpy,agrawal2018rewriting}, with an underlying off-the-shelf solver such as MOSEK~\cite{mosek}. However, there remains a further issue to be resolved: this typically requires the objective functions and constraints to be constructed from a given list of atomic functions in a \term{disciplined convex programming} (DCP) ruleset~\cite{liberti_disciplined_2006}. Unfortunately, the function $\EUR_\beta$ is not directly supported in such rulesets (though the KL divergence is).

We use the following approach to work around this issue: since $\EUR_\beta$ is a convex function, any tangent line to it yields a lower bound~\cite{boyd_convex_2004}, sometimes known as an \term{affine minorant}. Thus a lower bound of 
$\EUR_\beta(Q)$
can be obtained by considering a collection of tangent lines over any set $S$ of values of $Q$, namely
\begin{equation}
\label{eq:maxtangents}
\EUR_\beta(Q) \geq \sup_{Q^\star\in S} \, l_{Q^\star}(Q),
\end{equation}
where $l_{Q^\star}$ denotes the tangent line to $\EUR_\beta$ at $Q^\star$.
In other words, the pointwise maximum of these tangent lines is a valid lower bound to the original function.

The first derivative of $\EUR_\beta$ is 
\begin{equation}
\label{eq:derivative}
\EUR_\beta'(Q) = \frac{
    \beta(1-Q)^{\beta-1} - \beta Q^{\beta-1}}
    {(1-\beta)((1-Q)^\beta + Q^\beta)\ln2}.
\end{equation}
The equation of the tangent line can then be straightforwardly obtained from this derivative formula.

Furthermore, these tangent lines are valid atomic functions in CVXPY, which also supports taking their maximum under its DCP ruleset. Therefore, the keyrate computed using the RHS of~\eqref{eq:maxtangents} in place of $\EUR_\beta(Q)$ will be lower than the theoretical achievable keyrate established by (\ref{eq:keyrate}). Thus, by using convex optimisation\footnote{During the implementation of the convex optimisation, the convergence tolerance 
in CVXPY
was set to $10^{-9}$.}, any keyrate below the computed value is guaranteed to be secure and a certified achievable keyrate is found.

\subsection{Target protocol parameters}
\label{subsec:protparams}

For our subsequent demonstrations of the keyrate computations, we need to pick some example parameters we wish to achieve in the protocol. Table \ref{tab:parameter_values} shows the value we used for each of these parameters, to match those in~\cite{arqand_generalized_2025}. 

\begin{table}[ht]
\caption{Protocol parameter values used in our example computations, matching the values used in~\cite{arqand_generalized_2025}.}
\centering
\renewcommand{\arraystretch}{1.2}
\begin{tabular}{cc}
\toprule
\textbf{Parameter} & \textbf{Value} \\
\midrule
$\varepsilon_{\mathrm{secure}}$ & $10^{-10}$ \\
$Q_{\mathrm{thresh}}$ & $0.025$ \\
$\xi$ & $1.1$ \\
$\epsperp$ & $0.01$ \\
\bottomrule
\end{tabular}

\label{tab:parameter_values}
\end{table}

\begin{remark}
\label{rem:xiEC}
A caveat: while we picked the error-correction parameter $\xi=1.1$ in order to directly compare our findings to previous works, it was pointed out in~\cite{TMP+17} that when $n \lesssim 10^4$, this value may lead to impractically high abort probabilities in the error-correction step (in that the error-correction string is not long enough to allow Bob to guess Alice's string with high probability). Hence in our subsequent plots, the values in that regime should be understood under the caveat that they may correspond to significant probability of aborting during error correction --- we have indicated these regions in red in our plots. However, a more refined treatment would require more detailed simulation of the performance of such error-correcting codes, and we leave that beyond the scope of this work.
\end{remark}

\subsection{Parameters for maximizing keyrates}
\label{subsec:maxoverparams}

The above approach can thus be used to securely compute the keyrate for a given choice of $\gamma \in (0,1)$ and $\alpha \in (1,2)$ --- respectively, these are the test-round probability in the protocol, and an abstract parameter in~\eqref{eq:keyrate} describing a {\Renyi} entropy (see Appendix~\ref{app:sixstate}) used in the proof~\cite{arqand_generalized_2025}. Since \emph{any} choice of these parameters yields a valid  keyrate, we can maximize over their values to obtain a higher keyrate. To do so, we simply perform a grid search over the parameter space for these values. 

The intervals of parameter space that the grid search was performed on are $\gamma \in [10^{-2.5}, 10^0]$ and $\alpha \in [1+\frac{10^{-2}}{\sqrt{n}}, 1 + \frac{10^{2}}{\sqrt{n}}]$, on a logarithmic scale, based on some empirical observations in~\cite{arqand_generalized_2025}. We take the highest keyrate found in this process (for each $n$) as the value displayed in our subsequent plots. 

It remains to specify the grid size of this grid search, as well as the number of tangent lines to use in~\eqref{eq:maxtangents}. The grid size and number of tangent lines were decided after plotting the keyrate against each of the parameters. The grid size and number of tangent lines were chosen such that the keyrate begins to plateau at that value. This balances the computational cost and the optimality of the computed keyrate.

From Fig. \ref{fig:tangent_convergence} and \ref{fig:grid converge}, it was observed that the keyrate remains approximately constant when $n_\mathrm{tangent}\ge100$ and $n_{\mathrm{grid}}\ge50$. Thus, for subsequent calculation of the achievable keyrates, the number of tangent lines used will be $100$ and the grid size used will be $100\times100$ unless otherwise specified.

\begin{figure}[ht]
\centering
\includegraphics[width=0.4\textwidth]{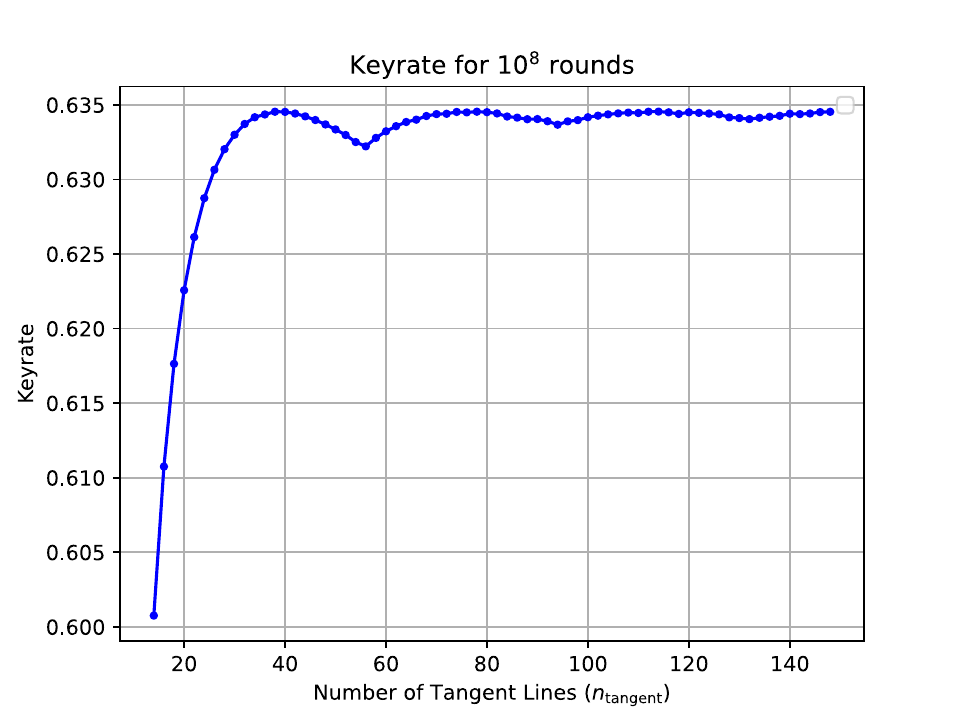} 
\caption{\textbf{Keyrate as a function of the number of tangent lines.
} The grid search used to plot this graph was $40 \times 40$. The plot is not monotone because for each $n_\mathrm{tangent}$ we took the tangent lines at uniformly spaced points over $Q\in[0,1]$, which results in the set of tangent lines for each $n_\mathrm{tangent}$ not being a subset of those for $n_\mathrm{tangent}+1$.}
\label{fig:tangent_convergence}
\end{figure}

\begin{figure}[ht]
\centering
\includegraphics[width=0.4\textwidth]{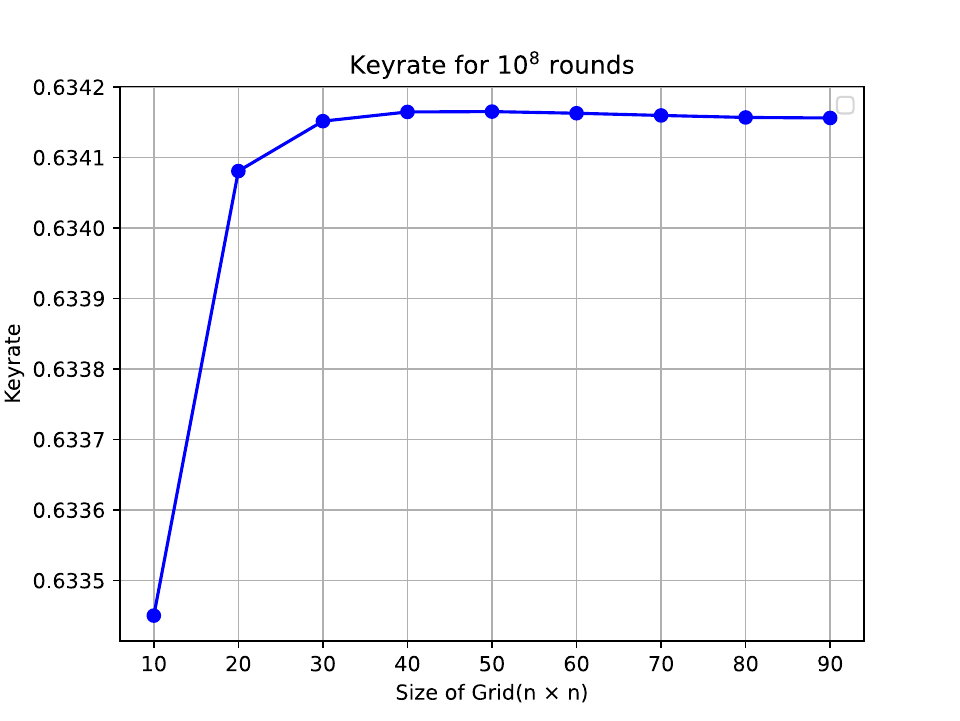} 
\caption{\textbf{Plot of keyrate value as the grid size of the grid search increases.} 100 tangent lines were used in the approximation of the convex optimisation.}
\label{fig:grid converge}
\end{figure}

\section{Results for BB84 protocol}
\label{sec:BB84results}

\subsection{Achievable Keyrate}

To provide a fair comparison with the results in~\cite{arqand_generalized_2025}, we first compute the achievable keyrate using the optimal $\alpha$ and $\gamma$ values from that work. Presumably, a certified achievable keyrate would be slightly lower than the value they computed, as the latter was based only on heuristic methods for evaluating the optimisation. From Fig. \ref{fig:keyrate compare} it was found that the achievable keyrate computed using our methodology is slightly lower than the value computed using those heuristic methods, validating that supposition. However, the difference is fairly small, showing that not much loss was incurred by rigorously certifying the value. 

\begin{figure}
\centering
\includegraphics[width=0.4\textwidth]{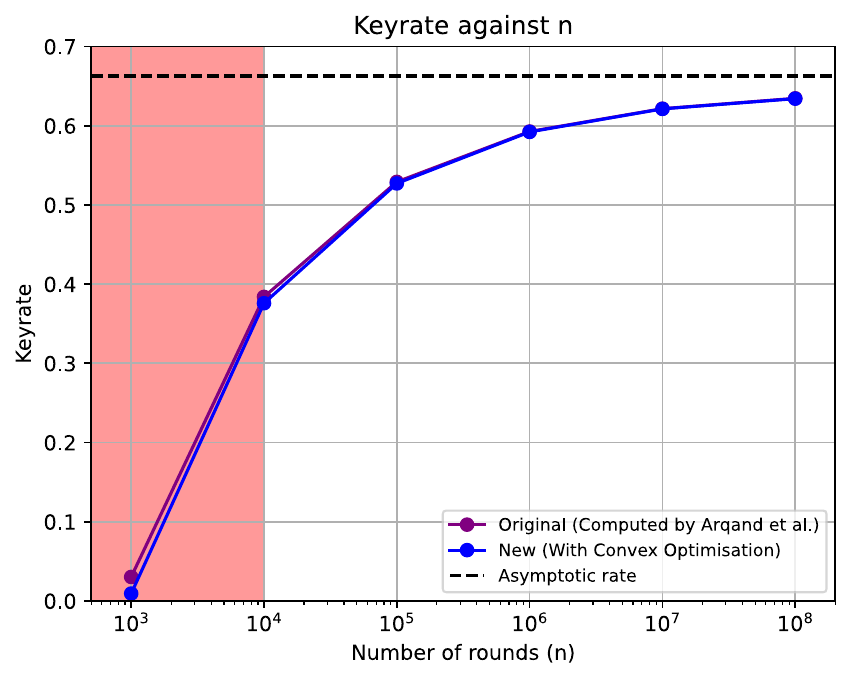} 
\caption{\textbf{Plot comparing the achievable keyrate computed by the original authors to that computed by our proposed methodology, using the values of $\alpha$ and $\gamma$ found in~\cite{arqand_generalized_2025}.} The purple graph plots the values computed by the original authors. The blue graph plots the values computed using the proposed methodology, with the values of $\alpha$ and $\gamma$ found in~\cite{arqand_generalized_2025}. The dashed black line shows the asymptotic keyrate for our parameter choices, which is $r_{\mathrm{asymp}}=0.663$. We caution however that the results displayed for $n\lesssim 10^4$ (red region in figure) may correspond to significant abort probability during error correction; see Remark~\ref{rem:xiEC}.}
\label{fig:keyrate compare}
\end{figure}

Next, to consider potential further improvements on the achievable keyrate, the optimal $\alpha$ and $\gamma$ was reevaluated using a $100\times100$ grid search. The achievable keyrate showed no significant increase from the results in Fig. \ref{fig:keyrate compare}. This observation suggests that~\cite{arqand_generalized_2025} already used a sufficiently fine grid to find nearly optimal values of $\gamma$ and $\alpha$.

To further improve on the achievable keyrate, 
we next computed the keyrates resulting from using the improved values of $\beta$ and $K_\alpha$ in~(\ref{eq:new parametrization}). There was a significant improvement in the keyrate at $n=10^3$ in comparison to when (\ref{eq:original parametrization}) was used as seen in Fig. \ref{fig:keyrate compare final}. 
However, the improvement at larger values of $n$ was fairly small. This indicates that this improved dependency on the {\Renyi} parameter $\alpha$ is mostly only significant at smaller values of $n$, in similar regimes to what was observed in previous work~\cite{ANB24}.

\begin{figure}
\centering
\includegraphics[width=0.4\textwidth]{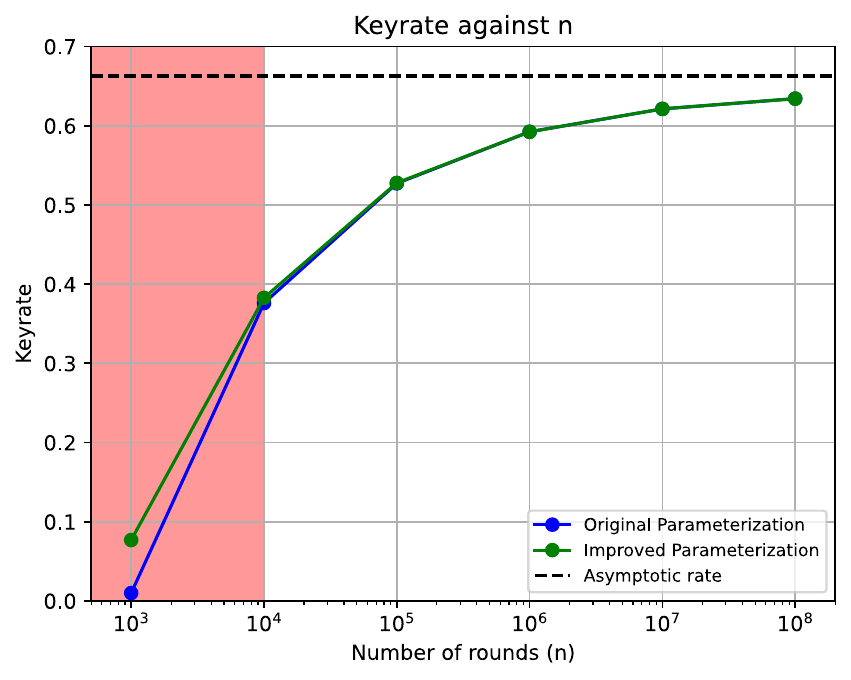} 
\caption{\textbf{Plot comparing the achievable keyrate computed by our proposed methodology, including the grid search to find the optimal value of $\gamma$ and $\alpha$, using different parametrisation of $\beta$ and $K_\alpha$.} The blue graph plots the values computed using (\ref{eq:original parametrization}). The green graph plots the values computed using (\ref{eq:new parametrization}). The dashed black line shows the asymptotic keyrate for our parameter choices, which is $r_{\mathrm{asymp}}=0.663$.  We caution however that the results displayed for $n\lesssim 10^4$ (red region in figure) may correspond to significant abort probability during error correction; see Remark~\ref{rem:xiEC}.}
\label{fig:keyrate compare final}
\end{figure}

Although our methodology was able to compute an achievable keyrate comparable to the original authors, we note that the evaluation of the convex optimisation problem sometimes encountered convergence issues. In particular, during the grid search, the convex optimisation was unable to converge for a significant fraction of the values of $\alpha$ and $\gamma$. Our keyrate plots were obtained by only considering $\alpha,\gamma$ values in which the optimisation converged successfully, thereby ensuring the keyrates were indeed rigorously certified up to solver accuracy; there were enough such points that we were still able to obtain reasonable keyrates in this fashion. Presumably, the achievable keyrate might increase slightly if the numerical stability could be improved to allow stable evaluation of more points in the grid.

\subsection{Optimal Parameters}

\begin{figure}[ht]
\centering
\includegraphics[width=0.4\textwidth]{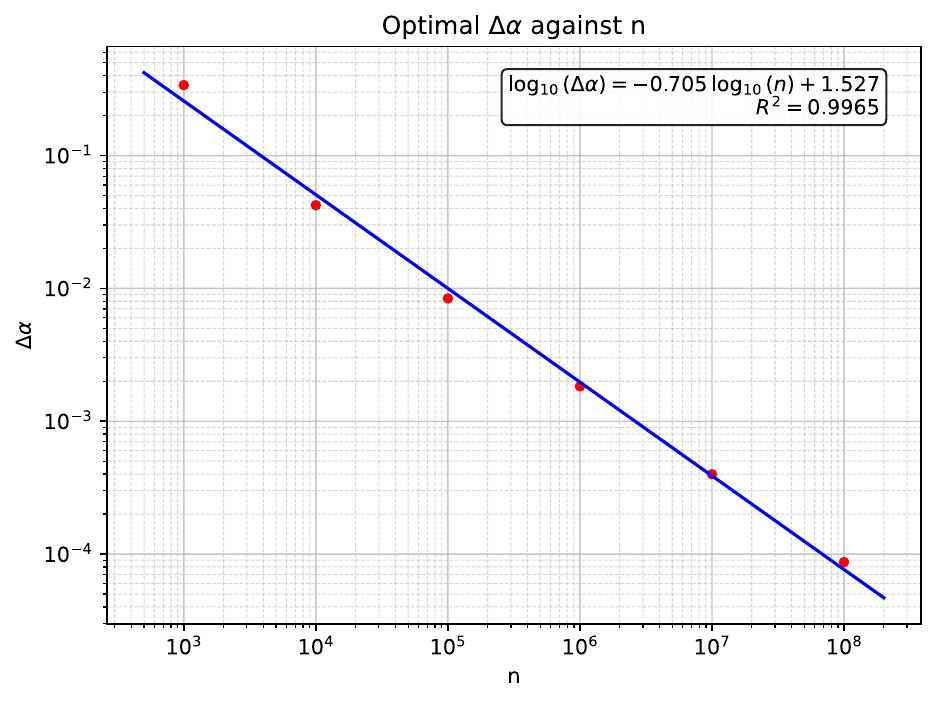} 
\caption{\textbf{Plot of the optimal $\Delta\alpha$ against the number of rounds.} Linear regression was performed on the values and a $R^2$-value of 0.9965 was found.}
\label{fig:optimal alpha}
\end{figure}

\begin{figure}
\centering
\includegraphics[width=0.4\textwidth]{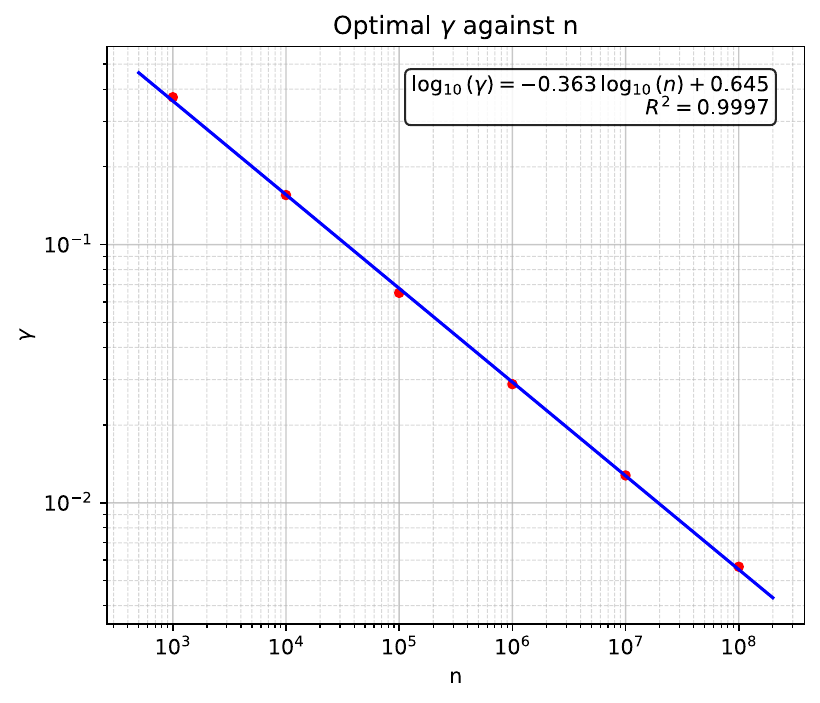} 
\caption{\textbf{Plot of the optimal $\gamma$ against the number of rounds.} Linear regression was performed on the values and a $R^2$-value of 0.9997 was found.}
\label{fig:optimal gamma}
\end{figure}

For the $\beta$ and $K_\alpha$ values in~(\ref{eq:new parametrization}), we plot the optimal $\alpha$ and $\gamma$ values we found in Fig. \ref{fig:optimal alpha} and \ref{fig:optimal gamma}. On a logarithmic scale, a strong linear trend was observed, with the best-fit lines as follows (defining $\Delta \alpha \coloneqq \alpha-1$): 

\begin{equation} 
\begin{aligned}
\log_{10}(\Delta\alpha) = -0.705 \log_{10}n + 1.527, \\
\log_{10}\gamma = -0.363 \cdot \log_{10}n + 0.645.
\end{aligned}
\end{equation}

Note that this scaling for $\gamma$ and $\alpha$ is based purely on empirical evidence. Despite the lack of mathematical basis, the establishment of this linear relationship should be beneficial for future implementation. As performing a grid search over the parameter space can be computationally expensive, this empirical trend can guide estimations of the optimal parameter. The range of the grid search can therefore be tightened, reducing the computational time of the grid search. 

To improve on our results, a possible goal for future work could be to derive these observed scaling factors rigorously, to provide stronger mathematical basis to the optimal scaling law. Such analysis has been performed based on other keyrate formulas for this protocol~\cite{hayashi_optimum_2022}. Thus, a similar approach could potentially be applied to (\ref{eq:keyrate}).

\section{Analysis of six-state protocol}
\label{sec:sixstateanalysis}

\newcommand{\instate}{\rho}
\newcommand{\outstate}{\nu}
\newcommand{\HSU}{\tilde{H}^\uparrow} 
\newcommand{\HSD}{\tilde{H}^\downarrow} 
\newcommand{\HPU}{\overline{H}^\uparrow} 
\newcommand{\ann}{I} 

\subsection{Protocol description}

We now turn to considering the six-state protocol.
For the purposes of this work, we consider an entanglement-based six-state protocol as follows. We highlight that there are many subtle differences in the exact definition used for ``the'' six-state protocol across different works, and hence our subsequent analysis should be understood as only being explicitly shown to hold for this specific form. 
\begin{algorithm}[H]
\caption{EB-six-state protocol outline.}
\label{prot:EB_sixstate}
\begin{algorithmic}[1]
\State For each round $j \in \{1,2,\dots,n\}$, perform the following steps:
\begin{algsubstates}
\State Alice and Bob each receive and briefly store a qubit register. Then via public communication, with probability $\gamma \in (0,1)$ (independently in each round) they jointly declare the round is a test round, and otherwise it is a generation round. Additionally, they publicly and jointly choose a classical value $\ann_j \in\{X,Y,Z\}$ uniformly at random, then Alice measures $\sigma_{\ann_j}$ and Bob measures $-\sigma_{\ann_j}$. 
\State If it is a generation round, they store the resulting values in registers $S_j$ and $\widetilde{S}_j$ for Alice and Bob respectively. Also, they jointly set a classical register $\CP_j=\perp$. 
\State If it is a test round, they publicly announce the resulting values, then jointly set a classical register $\CP_j=0$ if their outcomes matched and $\CP_j=1$ otherwise. Also, Alice sets $S_j=0$ and Bob sets $\widetilde{S}_j = \texttt{test}$.
\end{algsubstates}
\State All subsequent steps are identical to Protocol~\ref{prot:EB_BB84}, except for the choice of $\lkey$ which we shall describe below.
\end{algorithmic}
\end{algorithm}
Again, we have used $\sigma_X, \sigma_Y, \sigma_Z$ to denote the Pauli operators and corresponding qubit measurements. There is a technical implementation difference compared to Protocol~\ref{prot:EB_BB84}: for that protocol, the honest implementation was implicitly based on the Bell state $\Phi^+$ (see~\eqref{eq:bellstates} for explicit definitions of the Bell state notation), up to some noise tolerance induced by $Q_\mathrm{thresh}$. In contrast, in order for the honest implementation of Protocol~\ref{prot:EB_sixstate} to accept with high probability, it has to be based on $\Psi^-$ instead, again up to some tolerance induced by $Q_\mathrm{thresh}$. This is due to the technical point that there is no analogue of $\Psi^-$ that produces perfectly correlated instead of anticorrelated outcomes when Alice and Bob perform the same Pauli measurement. 

There are other possible conventions for the choices of states and measurements to address this difficulty --- for instance, we could instead have Bob always measure in the same basis as Alice and relabel the outcomes in some but not all bases. Still, they are all basically isomorphic up to outcome relabelling, so Protocol~\ref{prot:EB_sixstate} is the convention we shall use in this work.

\subsection{Keyrate formula}

To describe the keyrates, we shall have to make use of a quantity known as \term{sandwiched {\Renyi} entropy} $\HSU_\alpha$. This is a generalization of the von Neumann entropy, and has many useful properties --- however, we defer the definition to Appendix~\ref{app:sixstate}, as our discussion here will not require its detailed properties. 

Following the security proof in Ref.~\cite{arqand_generalized_2025}, one obtains the following formula for an achievable keyrate for the six-state protocol that we described:
\begin{equation}
\label{eq:keyrate6state}
\begin{aligned}
\lkey &= \inf_{\substack{
\nu \in \Sigma_{S\CP \ann E}\\
\mathbf{q} \in S_{\Omega}
}} 
\Bigg[
 q(\perp) \HSU_\alpha(S|\CP \ann E)_{\outstate_{|\perp}}
+ 
K_\alpha
D\!\left(\mbf{q} \parallel \boldsymbol{\nu}_{\CP}\right)
\Bigg] n 
\\
&\quad
- \lambda_{\mathrm{EC}}
- \log\!\left(\frac{1}{\varepsilon_{\mathrm{EV}}}\right)
- \frac{\alpha}{\alpha-1}
\log\!\left(\frac{1}{\varepsilon_{\mathrm{PA}}}\right)
+ 2, 
\end{aligned}
\end{equation}
where $\Sigma_{S\CP \ann E}$ is now the set of all states that could be produced on registers $S\CP \ann E$ in any single round of the protocol (omitting the round-$j$ subscripts for simplicity). Here, $\outstate_{|\perp}$ denotes conditioning the state $\nu \in \Sigma_{S\CP \ann E}$ on $\CP=\perp$, and we use $\boldsymbol{\nu}_{\CP}$ to denote the classical distribution induced by $\nu$ on $\CP$. 

The parameters and notation are defined in the same way as for~\eqref{eq:keyrate}, though for ease of presentation, when computing keyrates for the six-state protocol we only use the improved parameters in~\eqref{eq:new parametrization} (with this again being justified by Theorem~VI.1 in Ref.~\cite{arqand_generalized_2025}).
We keep the error-correction term $\lambda_{\mathrm{EC}}$ the same, under the assumption that the relevant noise model remains as depolarising noise. 

Basically, the BB84 keyrate formula in~\eqref{eq:keyrate} was obtained in Ref.~\cite{arqand_generalized_2025} by lower bounding the entropy term $\HSU_\alpha(S|\CP \ann E)$ (in the context of that protocol) with the function $\EUR_\beta$ presented in the main text. Our goal now is to obtain a similar function for the six-state protocol. As we show in Appendix \ref{app:sixstate}, for this protocol we can obtain the following lower bound on that entropy:
\begin{equation}
\label{eq:renyi used}
\begin{aligned}
\HSU_\alpha(S|\CP \ann E)_{\outstate_{|\perp}} &\geq \sixstate{\alpha}(Q) \text{ for } Q\in[0,2/3], \quad \text{where}\\
\sixstate{\alpha}(Q) &= \frac{\alpha}{1-\alpha}\log \Bigg[2^{1-\frac{1}{\alpha}}Q 
+ (1-Q)^{1-\frac{1}{\alpha}} \Bigg( \bigg( \frac{Q}{2} \bigg)^\frac{1}{\alpha} \\
&+ \bigg( 1 - \frac{3Q}{2} \bigg)^\frac{1}{\alpha} \Bigg) \Bigg] + 1,
\end{aligned}
\end{equation}
where $Q$ is related to $\boldsymbol{\nu}_{\CP}$ in the same way as in~\eqref{eq:nu_and_q}, and the restriction on its domain arises for technical reasons we discuss in Appendix \ref{app:sixstate}. With this in mind, the minimisation in~\eqref{eq:keyrate6state} can be lower bounded with
\begin{equation}
\label{eq:infoverQsixstate}
\begin{aligned}
\inf_{\substack{
Q\in[0,2/3]\\
\mathbf{q} \in S_{\Omega}
}} 
\Bigg[
 q(\perp) \sixstate{\alpha}(Q)
+ 
K_\alpha
D\!\left(\mbf{q} \parallel \boldsymbol{\nu}\right)
\Bigg] \\
\text{where } \boldsymbol{\nu} = (\gamma(1-Q), \gamma Q, 1-\gamma).
\end{aligned}
\end{equation}

\begin{remark}
\label{remark:renyis}
We obtained the bound $\sixstate{\alpha}$ by lower bounding $\HSU_\alpha$ with a different variant of {\Renyi} entropy, as we discuss in Appendix~\ref{app:sixstate}. This is the only point in our analysis that introduces looseness between~\eqref{eq:infoverQsixstate} and the original optimisation in~\eqref{eq:keyrate6state}; all other steps between those formulas were tight. 
In addition to this, we also considered approaches based on several other variants of {\Renyi} entropy and compared the resulting bounds; however, we found empirically that this version seemed to give the best performance in practice, hence we restrict our discussion here to this version.
We discuss the details further in that appendix. 
\end{remark}

As $n\to\infty$, one can take $\alpha\to 1$ at a suitable rate~\cite{arqand_generalized_2025} to obtain the standard formula for the asymptotic keyrate of the six-state protocol (as presented in e.g.~\cite{SBC+09}):
\begin{align}
r_{\mathrm{asymp}} &= 1 - Q_{\mathrm{thresh}} - (1 - Q_{\mathrm{thresh}}) h_{\mathrm{bin}}\left(\frac{1 - 3Q_{\mathrm{thresh}}/2}{1 - Q_{\mathrm{thresh}}}\right) \nonumber\\
&\qquad -h_{\mathrm{bin}}(Q_{\mathrm{thresh}}).
\end{align}

\subsection{Convex Optimisation}

We can evaluate the convex optimisation in~\eqref{eq:infoverQsixstate} using the same approach as we did for BB84, thereby obtaining a certified keyrate. To do so, we need the fact that $\sixstate{\alpha}$ is convex, which we show in Appendix \ref{app:sixstate}, and we also need its first derivative, which is as follows:
\begin{widetext}
\begin{equation}
\label{eq:first renyi diff}
\sixstate{\alpha}'(Q) = \frac{\alpha}{1-\alpha} \left(
\frac{
2^{1-\frac{1}{\alpha}}
-\left(1-\frac{1}{\alpha}\right)
(1-Q)^{-\frac{1}{\alpha}}
\left[
\left(\frac{Q}{2}\right)^{\frac{1}{\alpha}}
+
\left(1-\frac{3Q}{2}\right)^{\frac{1}{\alpha}}
\right]
+ (1-Q)^{1-\frac{1}{\alpha}}
\left[
\frac{1}{2\alpha}
\left(\frac{Q}{2}\right)^{\frac{1}{\alpha}-1}
-
\frac{3}{2\alpha}
\left(1-\frac{3Q}{2}\right)^{\frac{1}{\alpha}-1}
\right]
}{\displaystyle \Bigg(
2^{1-\frac{1}{\alpha}}Q
+
(1-Q)^{1-\frac{1}{\alpha}}
\left[
\left(\frac{Q}{2}\right)^{\frac{1}{\alpha}}
+
\left(1-\frac{3Q}{2}\right)^{\frac{1}{\alpha}}
\right] \Bigg) \ln 2
} \right)
\end{equation}
\end{widetext}

Using (\ref{eq:first renyi diff}), we can find the tangent lines of~$\sixstate{\alpha}$, and implement the resulting convex optimisation in CVXPY the same way we did for BB84.
We briefly note a technical point here: as mentioned in Remark~\ref{remark:renyis}, we in fact obtained several different bounds on $\HSU_\alpha(S|\CP \ann E)$, some of which were tighter than others. However, if we take the tangent lines at a fixed finite set of points, the resulting piecewise-linear bounds obtained by maximising over these tangent lines do not necessarily maintain the same ordering as the original bounds, as the tangent lines may intersect differently from the original functions. We discuss the details in Appendix~\ref{subsec:boundperformance}.

\section{Results for six-state protocol}
\label{sec:sixstateresults}

\subsection{Achievable Keyrate}

Using the same parameters as shown in Table \ref{tab:parameter_values}, the achievable keyrate for the six-state protocol was computed. As shown in the parameter tuning for the EB-BB84 protocol, $100$ tangent lines and a grid size of $100 \times 100$ is sufficient for the implementation and will be used for the computation of the achievable keyrate for the six-state protocol.

\begin{figure}[ht]
\centering
\includegraphics[width=0.4\textwidth]{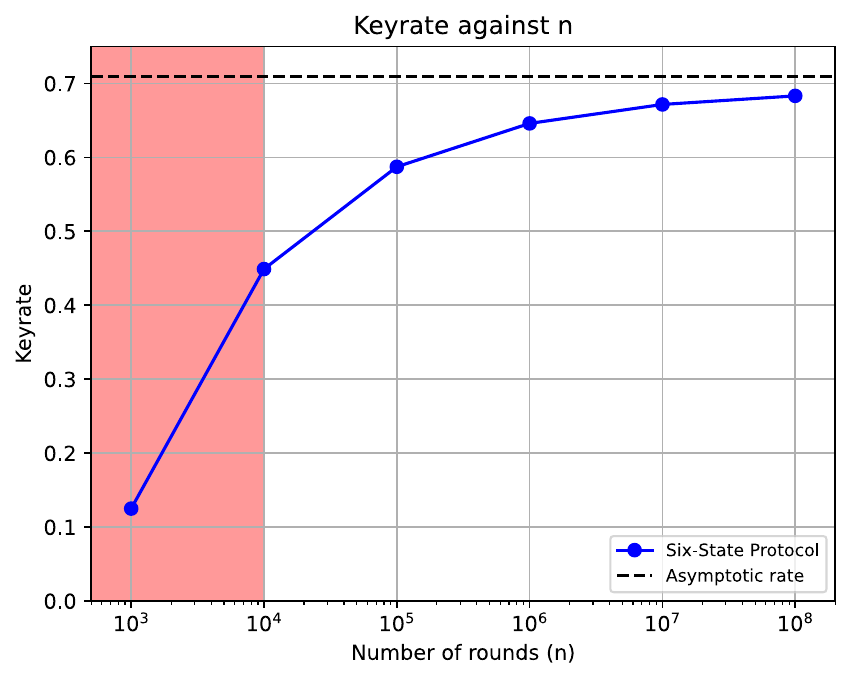} 
\caption{\textbf{Plot of the achievable keyrate computed by our proposed methodology, for the six-state protocol.} The blue graph plots the achievable keyrate of the protocol. 
The dashed black line shows the asymptotic keyrate for our parameter choices, which is $r_{\mathrm{asymp}}=0.710$.  We caution however that the results displayed for $n\lesssim 10^4$ (red region in figure) may correspond to significant abort probability during error correction; see Remark~\ref{rem:xiEC}.
}
\label{fig:keyrate sixstate}
\end{figure}

The keyrates for the six-state protocol are higher than those for BB84, as expected \cite{renner_security_2008}.

\subsection{Optimal Parameters}

\begin{figure}[ht]
\centering
\includegraphics[width=0.4\textwidth]{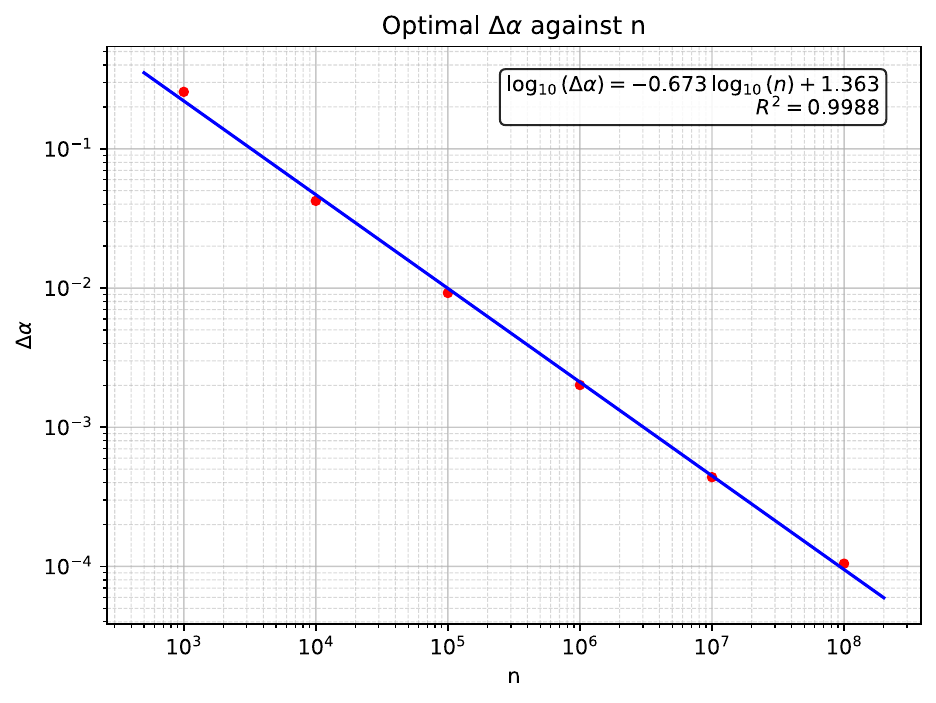} 
\caption{\textbf{Plot of the optimal $\Delta\alpha$ against the number of rounds.} Linear regression was performed on the values and a $R^2$-value of 0.9988 was found.}
\label{fig:optimal alpha ss}
\end{figure}

\begin{figure}[ht]
\centering
\includegraphics[width=0.4\textwidth]{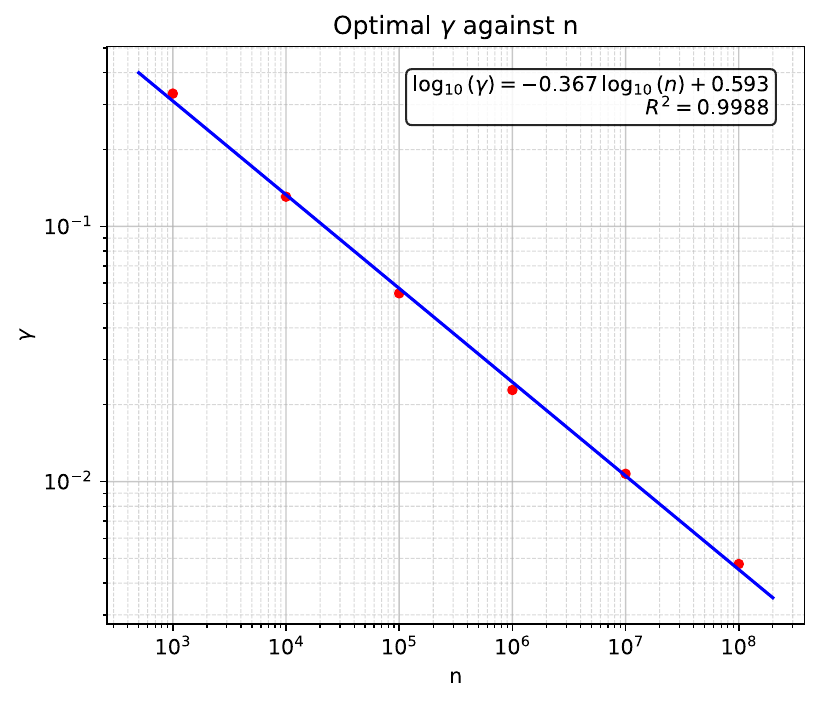} 
\caption{\textbf{Plot of the optimal $\gamma$ against the number of rounds.} Linear regression was performed on the values and a 
$R^2$-value 
of 0.9988 was found.}
\label{fig:optimal gamma ss}
\end{figure}

When the optimal logarithmic values of $\alpha$ and $\gamma$ were plotted against the number of rounds for the six-state protocol, a strong linear trend was similarly found.
The following are the equations of the best fit lines for the optimal $\alpha$ and $\gamma$ parameters respectively in the six-state protocol, defining $\Delta \alpha \coloneqq \alpha-1$:

\begin{equation} 
\begin{aligned}
     \log_{10}(\Delta\alpha) = -0.673 \log_{10}n + 1.363, \\
     \log_{10}\gamma = -0.367 \cdot \log_{10}n + 0.593.
\end{aligned}
\end{equation}

Despite the BB84 and six-state protocols being two different protocols, the equations of the best fit lines for the optimal $\alpha$ and $\gamma$ are similar.

\section{Conclusion}
\label{sec:conclusion}

This work outlines a methodology to compute a certified finite-size keyrate of basic QKD protocols using off-the-shelf convex optimisation solvers. By constructing suitable tangent lines to lower bound the original objective function, the optimisation can be formulated as a convex problem that off-the-shelf solvers can handle. This provides an alternative to the sophisticated optimisation framework developed in recent years for analysing a broad range of QKD protocols.

We first applied this methodology to the EB-BB84 protocol to compute its achievable keyrate. However, it was found that the first term in (\ref{eq:infoverQ}) was not convex. Thus, a convex lower bound was derived prior to the implementation of convex optimisation. To maximise the achievable keyrate, a grid search was also performed on the parameter space of $\gamma$ and $\alpha$ to determine the optimal choices of these parameters.
Using these parameters, the achievable keyrate for EB-BB84 was computed, showing only a slight decrease from the results in~\cite{arqand_generalized_2025} previously computed using heuristic methods. 

The achievable keyrate was further improved by adopting an improved dependency on the {\Renyi} parameter as compared to that used in~\cite{arqand_generalized_2025} (namely, (\ref{eq:new parametrization}) in place of (\ref{eq:original parametrization})). Our results demonstrate that our methodology can be used to evaluate the achievable keyrate for the EB-BB84 protocol using standard convex optimisation tools, without significantly decreasing the keyrate.

Following this, we extended the approach to the six-state protocol. To enable us to find the keyrate using the entropy accumulation framework, we needed a closed-form lower bound on the single-round {\Renyi} entropy. By analysing three different versions of {\Renyi} entropy, we derived suitable lower bounds. Based on empirical evidence, we found that \eqref{eq:sixstatebnd} yielded the best performance among the three expressions. Thus, \eqref{eq:sixstatebnd} was used to compute the achievable keyrate for the six-state protocol. The achievable keyrate across different number of rounds was found and shown in Fig. \ref{fig:keyrate sixstate}. These results form a basis for further finite-size analysis of the six-state protocol using entropy accumulation.

For both the EB-BB84 and the six-state protocol, empirical evidence shows a decreasing linear trend between the optimal parameters of $\gamma$ and $\alpha$ as the number of rounds, $n$, increased in the logarithmic scale. Additionally, the corresponding equations of the best fit line for both protocols are similar. While this behaviour is only supported by empirical evidence, a mathematical derivation may provide useful insight into the optimisation of these parameters for finite-size QKD protocols.

In conclusion, our results demonstrate that a certified finite-size keyrate can be computed using readily available convex optimisation solvers, while retaining competitive performance. The extension to the six-state protocol, together with the derived single-round {\Renyi} entropy expressions, broadens the applicability of this methodology and facilitates subsequent studies of the protocol within the entropy accumulation framework. Future work could extend our methodology to other variants of BB84, such as the decoy-state BB84 protocol \cite{Hwang03,lo_decoy_2005,MQZL05,Wang05}. \\

\noindent\textbf{Data Availability Statement:} All code and data used in this project can be found in the following \href{https://github.com/HackoAwesome/BB84keyrate}{GitHub Repository}.

\begin{acknowledgments}
We thank Hayley Lim Hui En for her contributions in the initial phase of this project.
We would also like to express our sincere gratitude towards the Special Programme in Science faculty members and our mentors, Jin Jiarui and Foo Tun Min, for their continuous support and constructive feedback for our project. This project would not have been possible without their help.
\end{acknowledgments}

\bibliography{references_bibtex}

\begin{thebibliography}{29}%
\makeatletter
\providecommand \@ifxundefined [1]{%
 \@ifx{#1\undefined}
}%
\providecommand \@ifnum [1]{%
 \ifnum #1\expandafter \@firstoftwo
 \else \expandafter \@secondoftwo
 \fi
}%
\providecommand \@ifx [1]{%
 \ifx #1\expandafter \@firstoftwo
 \else \expandafter \@secondoftwo
 \fi
}%
\providecommand \natexlab [1]{#1}%
\providecommand \enquote  [1]{``#1''}%
\providecommand \bibnamefont  [1]{#1}%
\providecommand \bibfnamefont [1]{#1}%
\providecommand \citenamefont [1]{#1}%
\providecommand \href@noop [0]{\@secondoftwo}%
\providecommand \href [0]{\begingroup \@sanitize@url \@href}%
\providecommand \@href[1]{\@@startlink{#1}\@@href}%
\providecommand \@@href[1]{\endgroup#1\@@endlink}%
\providecommand \@sanitize@url [0]{\catcode `\\12\catcode `\$12\catcode `\&12\catcode `\#12\catcode `\^12\catcode `\_12\catcode `\%12\relax}%
\providecommand \@@startlink[1]{}%
\providecommand \@@endlink[0]{}%
\providecommand \url  [0]{\begingroup\@sanitize@url \@url }%
\providecommand \@url [1]{\endgroup\@href {#1}{\urlprefix }}%
\providecommand \urlprefix  [0]{URL }%
\providecommand \Eprint [0]{\href }%
\providecommand \doibase [0]{http://dx.doi.org/}%
\providecommand \selectlanguage [0]{\@gobble}%
\providecommand \bibinfo  [0]{\@secondoftwo}%
\providecommand \bibfield  [0]{\@secondoftwo}%
\providecommand \translation [1]{[#1]}%
\providecommand \BibitemOpen [0]{}%
\providecommand \bibitemStop [0]{}%
\providecommand \bibitemNoStop [0]{.\EOS\space}%
\providecommand \EOS [0]{\spacefactor3000\relax}%
\providecommand \BibitemShut  [1]{\csname bibitem#1\endcsname}%
\let\auto@bib@innerbib\@empty
\bibitem [{\citenamefont {Arqand}\ \emph {et~al.}(2025)\citenamefont {Arqand}, \citenamefont {Hahn},\ and\ \citenamefont {Tan}}]{arqand_generalized_2025}%
  \BibitemOpen
  \bibfield  {author} {\bibinfo {author} {\bibfnamefont {A.}~\bibnamefont {Arqand}}, \bibinfo {author} {\bibfnamefont {T.~A.}\ \bibnamefont {Hahn}}, \ and\ \bibinfo {author} {\bibfnamefont {E.~Y.-Z.}\ \bibnamefont {Tan}},\ }\href {\doibase 10.1103/pgrn-mz9j} {\bibfield  {journal} {\bibinfo  {journal} {Physical Review X}\ }\textbf {\bibinfo {volume} {15}},\ \bibinfo {pages} {041013} (\bibinfo {year} {2025})}\BibitemShut {NoStop}%
\bibitem [{\citenamefont {Tomamichel}(2016)}]{tomamichel_quantum_2016}%
  \BibitemOpen
  \bibfield  {author} {\bibinfo {author} {\bibfnamefont {M.}~\bibnamefont {Tomamichel}},\ }\href {\doibase 10.1007/978-3-319-21891-5} {\emph {\bibinfo {title} {Quantum Information Processing with Finite Resources}}},\ \bibinfo {series} {{SpringerBriefs} in Mathematical Physics}, Vol.~\bibinfo {volume} {5}\ (\bibinfo  {publisher} {Springer International Publishing},\ \bibinfo {address} {Cham},\ \bibinfo {year} {2016})\BibitemShut {NoStop}%
\bibitem [{\citenamefont {He}\ \emph {et~al.}(2024)\citenamefont {He}, \citenamefont {Saunderson},\ and\ \citenamefont {Fawzi}}]{he_qics_2024}%
  \BibitemOpen
  \bibfield  {author} {\bibinfo {author} {\bibfnamefont {K.}~\bibnamefont {He}}, \bibinfo {author} {\bibfnamefont {J.}~\bibnamefont {Saunderson}}, \ and\ \bibinfo {author} {\bibfnamefont {H.}~\bibnamefont {Fawzi}},\ }\href {\doibase 10.48550/ARXIV.2410.17803} {\enquote {\bibinfo {title} {{QICS}: Quantum information conic solver},}\ } (\bibinfo {year} {2024})\BibitemShut {NoStop}%
\bibitem [{\citenamefont {He}\ \emph {et~al.}(2026)\citenamefont {He}, \citenamefont {Saunderson},\ and\ \citenamefont {Fawzi}}]{he_operator_2026}%
  \BibitemOpen
  \bibfield  {author} {\bibinfo {author} {\bibfnamefont {K.}~\bibnamefont {He}}, \bibinfo {author} {\bibfnamefont {J.}~\bibnamefont {Saunderson}}, \ and\ \bibinfo {author} {\bibfnamefont {H.}~\bibnamefont {Fawzi}},\ }\href {\doibase 10.1007/s10107-025-02314-0} {\bibfield  {journal} {\bibinfo  {journal} {Mathematical Programming}\ } (\bibinfo {year} {2026}),\ 10.1007/s10107-025-02314-0}\BibitemShut {NoStop}%
\bibitem [{\citenamefont {Kamin}\ \emph {et~al.}(2025)\citenamefont {Kamin}, \citenamefont {Burniston},\ and\ \citenamefont {Tan}}]{kamin_renyi_2025}%
  \BibitemOpen
  \bibfield  {author} {\bibinfo {author} {\bibfnamefont {L.}~\bibnamefont {Kamin}}, \bibinfo {author} {\bibfnamefont {J.}~\bibnamefont {Burniston}}, \ and\ \bibinfo {author} {\bibfnamefont {E.~Y.~Z.}\ \bibnamefont {Tan}},\ }\href {\doibase 10.48550/ARXIV.2504.12248} {\enquote {\bibinfo {title} {Rényi security framework against coherent attacks applied to decoy-state {QKD}},}\ } (\bibinfo {year} {2025})\BibitemShut {NoStop}%
\bibitem [{\citenamefont {Navarro}\ \emph {et~al.}(2025)\citenamefont {Navarro}, \citenamefont {Lorente}, \citenamefont {Parellada}, \citenamefont {Pascual-García},\ and\ \citenamefont {Araújo}}]{navarro_finite-size_2025}%
  \BibitemOpen
  \bibfield  {author} {\bibinfo {author} {\bibfnamefont {M.}~\bibnamefont {Navarro}}, \bibinfo {author} {\bibfnamefont {A.~G.}\ \bibnamefont {Lorente}}, \bibinfo {author} {\bibfnamefont {P.~V.}\ \bibnamefont {Parellada}}, \bibinfo {author} {\bibfnamefont {C.}~\bibnamefont {Pascual-García}}, \ and\ \bibinfo {author} {\bibfnamefont {M.}~\bibnamefont {Araújo}},\ }\href {\doibase 10.48550/ARXIV.2511.10584} {\enquote {\bibinfo {title} {Finite-size quantum key distribution rates from {R}ényi entropies using conic optimization},}\ } (\bibinfo {year} {2025})\BibitemShut {NoStop}%
\bibitem [{\citenamefont {Boyd}\ and\ \citenamefont {Vandenberghe}(2004)}]{boyd_convex_2004}%
  \BibitemOpen
  \bibfield  {author} {\bibinfo {author} {\bibfnamefont {S.}~\bibnamefont {Boyd}}\ and\ \bibinfo {author} {\bibfnamefont {L.}~\bibnamefont {Vandenberghe}},\ }\href {\doibase 10.1017/CBO9780511804441} {\emph {\bibinfo {title} {Convex Optimization}}},\ \bibinfo {edition} {1st}\ ed.\ (\bibinfo  {publisher} {Cambridge University Press},\ \bibinfo {year} {2004})\BibitemShut {NoStop}%
\bibitem [{\citenamefont {Anco}\ \emph {et~al.}(2024)\citenamefont {Anco}, \citenamefont {Nemoz},\ and\ \citenamefont {Brown}}]{ANB24}%
  \BibitemOpen
  \bibfield  {author} {\bibinfo {author} {\bibfnamefont {K.~G.}\ \bibnamefont {Anco}}, \bibinfo {author} {\bibfnamefont {T.}~\bibnamefont {Nemoz}}, \ and\ \bibinfo {author} {\bibfnamefont {P.}~\bibnamefont {Brown}},\ }\href {https://arxiv.org/abs/2410.16447} {\enquote {\bibinfo {title} {{How much secure randomness is in a quantum state?}}}\ } (\bibinfo {year} {2024}),\ \Eprint {http://arxiv.org/abs/2410.16447} {arXiv:2410.16447 [quant-ph]} \BibitemShut {NoStop}%
\bibitem [{\citenamefont {Tomamichel}\ and\ \citenamefont {Leverrier}(2017)}]{tomamichel_largely_2017}%
  \BibitemOpen
  \bibfield  {author} {\bibinfo {author} {\bibfnamefont {M.}~\bibnamefont {Tomamichel}}\ and\ \bibinfo {author} {\bibfnamefont {A.}~\bibnamefont {Leverrier}},\ }\href {\doibase 10.22331/q-2017-07-14-14} {\bibfield  {journal} {\bibinfo  {journal} {Quantum}\ }\textbf {\bibinfo {volume} {1}},\ \bibinfo {pages} {14} (\bibinfo {year} {2017})}\BibitemShut {NoStop}%
\bibitem [{\citenamefont {Portmann}\ and\ \citenamefont {Renner}(2022)}]{PR22}%
  \BibitemOpen
  \bibfield  {author} {\bibinfo {author} {\bibfnamefont {C.}~\bibnamefont {Portmann}}\ and\ \bibinfo {author} {\bibfnamefont {R.}~\bibnamefont {Renner}},\ }\href {\doibase 10.1103/RevModPhys.94.025008} {\bibfield  {journal} {\bibinfo  {journal} {Reviews of Modern Physics}\ }\textbf {\bibinfo {volume} {94}},\ \bibinfo {pages} {025008} (\bibinfo {year} {2022})}\BibitemShut {NoStop}%
\bibitem [{\citenamefont {Cover}\ and\ \citenamefont {Thomas}(2005)}]{cover_elements_2005}%
  \BibitemOpen
  \bibfield  {author} {\bibinfo {author} {\bibfnamefont {T.~M.}\ \bibnamefont {Cover}}\ and\ \bibinfo {author} {\bibfnamefont {J.~A.}\ \bibnamefont {Thomas}},\ }\href {\doibase 10.1002/047174882X} {\emph {\bibinfo {title} {Elements of Information Theory}}},\ \bibinfo {edition} {1st}\ ed.\ (\bibinfo  {publisher} {Wiley},\ \bibinfo {year} {2005})\BibitemShut {NoStop}%
\bibitem [{\citenamefont {Tomamichel}\ \emph {et~al.}(2017)\citenamefont {Tomamichel}, \citenamefont {Martinez-Mateo}, \citenamefont {Pacher},\ and\ \citenamefont {Elkouss}}]{TMP+17}%
  \BibitemOpen
  \bibfield  {author} {\bibinfo {author} {\bibfnamefont {M.}~\bibnamefont {Tomamichel}}, \bibinfo {author} {\bibfnamefont {J.}~\bibnamefont {Martinez-Mateo}}, \bibinfo {author} {\bibfnamefont {C.}~\bibnamefont {Pacher}}, \ and\ \bibinfo {author} {\bibfnamefont {D.}~\bibnamefont {Elkouss}},\ }\href {\doibase 10.1007/s11128-017-1709-5} {\bibfield  {journal} {\bibinfo  {journal} {Quantum Information Processing}\ }\textbf {\bibinfo {volume} {16}} (\bibinfo {year} {2017}),\ 10.1007/s11128-017-1709-5}\BibitemShut {NoStop}%
\bibitem [{\citenamefont {Scarani}\ \emph {et~al.}(2009)\citenamefont {Scarani}, \citenamefont {Bechmann-Pasquinucci}, \citenamefont {Cerf}, \citenamefont {Du\ifmmode~\check{s}\else \v{s}\fi{}ek}, \citenamefont {L\"utkenhaus},\ and\ \citenamefont {Peev}}]{SBC+09}%
  \BibitemOpen
  \bibfield  {author} {\bibinfo {author} {\bibfnamefont {V.}~\bibnamefont {Scarani}}, \bibinfo {author} {\bibfnamefont {H.}~\bibnamefont {Bechmann-Pasquinucci}}, \bibinfo {author} {\bibfnamefont {N.~J.}\ \bibnamefont {Cerf}}, \bibinfo {author} {\bibfnamefont {M.}~\bibnamefont {Du\ifmmode~\check{s}\else \v{s}\fi{}ek}}, \bibinfo {author} {\bibfnamefont {N.}~\bibnamefont {L\"utkenhaus}}, \ and\ \bibinfo {author} {\bibfnamefont {M.}~\bibnamefont {Peev}},\ }\href {\doibase 10.1103/RevModPhys.81.1301} {\bibfield  {journal} {\bibinfo  {journal} {Reviews of Modern Physics}\ }\textbf {\bibinfo {volume} {81}},\ \bibinfo {pages} {1301} (\bibinfo {year} {2009})}\BibitemShut {NoStop}%
\bibitem [{\citenamefont {Blitzstein}\ and\ \citenamefont {Hwang}(2014)}]{blitzstein_introduction_2014}%
  \BibitemOpen
  \bibfield  {author} {\bibinfo {author} {\bibfnamefont {J.~K.}\ \bibnamefont {Blitzstein}}\ and\ \bibinfo {author} {\bibfnamefont {J.}~\bibnamefont {Hwang}},\ }\href {\doibase 10.1201/b17221} {\emph {\bibinfo {title} {Introduction to Probability}}}\ (\bibinfo  {publisher} {Chapman and Hall/{CRC}},\ \bibinfo {year} {2014})\BibitemShut {NoStop}%
\bibitem [{\citenamefont {Diamond}\ and\ \citenamefont {Boyd}(2016)}]{diamond2016cvxpy}%
  \BibitemOpen
  \bibfield  {author} {\bibinfo {author} {\bibfnamefont {S.}~\bibnamefont {Diamond}}\ and\ \bibinfo {author} {\bibfnamefont {S.}~\bibnamefont {Boyd}},\ }\href@noop {} {\bibfield  {journal} {\bibinfo  {journal} {Journal of Machine Learning Research}\ }\textbf {\bibinfo {volume} {17}},\ \bibinfo {pages} {1} (\bibinfo {year} {2016})}\BibitemShut {NoStop}%
\bibitem [{\citenamefont {Agrawal}\ \emph {et~al.}(2018)\citenamefont {Agrawal}, \citenamefont {Verschueren}, \citenamefont {Diamond},\ and\ \citenamefont {Boyd}}]{agrawal2018rewriting}%
  \BibitemOpen
  \bibfield  {author} {\bibinfo {author} {\bibfnamefont {A.}~\bibnamefont {Agrawal}}, \bibinfo {author} {\bibfnamefont {R.}~\bibnamefont {Verschueren}}, \bibinfo {author} {\bibfnamefont {S.}~\bibnamefont {Diamond}}, \ and\ \bibinfo {author} {\bibfnamefont {S.}~\bibnamefont {Boyd}},\ }\href@noop {} {\bibfield  {journal} {\bibinfo  {journal} {Journal of Control and Decision}\ }\textbf {\bibinfo {volume} {5}},\ \bibinfo {pages} {42} (\bibinfo {year} {2018})}\BibitemShut {NoStop}%
\bibitem [{\citenamefont {ApS}(2025)}]{mosek}%
  \BibitemOpen
  \bibfield  {author} {\bibinfo {author} {\bibfnamefont {M.}~\bibnamefont {ApS}},\ }\href {https://docs.mosek.com/latest/pythonfusion/index.html} {\emph {\bibinfo {title} {The MOSEK Python Fusion API manual. Version 11.1.}}} (\bibinfo {year} {2025})\BibitemShut {NoStop}%
\bibitem [{\citenamefont {Grant}\ \emph {et~al.}(2006)\citenamefont {Grant}, \citenamefont {Boyd},\ and\ \citenamefont {Ye}}]{liberti_disciplined_2006}%
  \BibitemOpen
  \bibfield  {author} {\bibinfo {author} {\bibfnamefont {M.}~\bibnamefont {Grant}}, \bibinfo {author} {\bibfnamefont {S.}~\bibnamefont {Boyd}}, \ and\ \bibinfo {author} {\bibfnamefont {Y.}~\bibnamefont {Ye}},\ }in\ \href {\doibase 10.1007/0-387-30528-9_7} {\emph {\bibinfo {booktitle} {Global Optimization}}},\ Vol.~\bibinfo {volume} {84},\ \bibinfo {editor} {edited by\ \bibinfo {editor} {\bibfnamefont {L.}~\bibnamefont {Liberti}}\ and\ \bibinfo {editor} {\bibfnamefont {N.}~\bibnamefont {Maculan}}}\ (\bibinfo  {publisher} {Kluwer Academic Publishers},\ \bibinfo {address} {Boston},\ \bibinfo {year} {2006})\ pp.\ \bibinfo {pages} {155--210}\BibitemShut {NoStop}%
\bibitem [{\citenamefont {Hayashi}(2022)}]{hayashi_optimum_2022}%
  \BibitemOpen
  \bibfield  {author} {\bibinfo {author} {\bibfnamefont {M.}~\bibnamefont {Hayashi}},\ }\href {\doibase 10.1103/PhysRevA.105.042603} {\bibfield  {journal} {\bibinfo  {journal} {Physical Review A}\ }\textbf {\bibinfo {volume} {105}},\ \bibinfo {pages} {042603} (\bibinfo {year} {2022})}\BibitemShut {NoStop}%
\bibitem [{\citenamefont {Renner}(2008)}]{renner_security_2008}%
  \BibitemOpen
  \bibfield  {author} {\bibinfo {author} {\bibfnamefont {R.}~\bibnamefont {Renner}},\ }\href {\doibase 10.1142/S0219749908003256} {\bibfield  {journal} {\bibinfo  {journal} {International Journal of Quantum Information}\ }\textbf {\bibinfo {volume} {06}},\ \bibinfo {pages} {1} (\bibinfo {year} {2008})}\BibitemShut {NoStop}%
\bibitem [{\citenamefont {Hwang}(2003)}]{Hwang03}%
  \BibitemOpen
  \bibfield  {author} {\bibinfo {author} {\bibfnamefont {W.-Y.}\ \bibnamefont {Hwang}},\ }\href {\doibase 10.1103/PhysRevLett.91.057901} {\bibfield  {journal} {\bibinfo  {journal} {Physical Review Letters}\ }\textbf {\bibinfo {volume} {91}},\ \bibinfo {pages} {057901} (\bibinfo {year} {2003})}\BibitemShut {NoStop}%
\bibitem [{\citenamefont {Lo}\ \emph {et~al.}(2005)\citenamefont {Lo}, \citenamefont {Ma},\ and\ \citenamefont {Chen}}]{lo_decoy_2005}%
  \BibitemOpen
  \bibfield  {author} {\bibinfo {author} {\bibfnamefont {H.-K.}\ \bibnamefont {Lo}}, \bibinfo {author} {\bibfnamefont {X.}~\bibnamefont {Ma}}, \ and\ \bibinfo {author} {\bibfnamefont {K.}~\bibnamefont {Chen}},\ }\href {\doibase 10.1103/PhysRevLett.94.230504} {\bibfield  {journal} {\bibinfo  {journal} {Physical Review Letters}\ }\textbf {\bibinfo {volume} {94}},\ \bibinfo {pages} {230504} (\bibinfo {year} {2005})}\BibitemShut {NoStop}%
\bibitem [{\citenamefont {Ma}\ \emph {et~al.}(2005)\citenamefont {Ma}, \citenamefont {Qi}, \citenamefont {Zhao},\ and\ \citenamefont {Lo}}]{MQZL05}%
  \BibitemOpen
  \bibfield  {author} {\bibinfo {author} {\bibfnamefont {X.}~\bibnamefont {Ma}}, \bibinfo {author} {\bibfnamefont {B.}~\bibnamefont {Qi}}, \bibinfo {author} {\bibfnamefont {Y.}~\bibnamefont {Zhao}}, \ and\ \bibinfo {author} {\bibfnamefont {H.-K.}\ \bibnamefont {Lo}},\ }\href {\doibase 10.1103/PhysRevA.72.012326} {\bibfield  {journal} {\bibinfo  {journal} {Physical Review A}\ }\textbf {\bibinfo {volume} {72}},\ \bibinfo {pages} {012326} (\bibinfo {year} {2005})}\BibitemShut {NoStop}%
\bibitem [{\citenamefont {Wang}(2005)}]{Wang05}%
  \BibitemOpen
  \bibfield  {author} {\bibinfo {author} {\bibfnamefont {X.-B.}\ \bibnamefont {Wang}},\ }\href {\doibase 10.1103/PhysRevLett.94.230503} {\bibfield  {journal} {\bibinfo  {journal} {Physical Review Letters}\ }\textbf {\bibinfo {volume} {94}},\ \bibinfo {pages} {230503} (\bibinfo {year} {2005})}\BibitemShut {NoStop}%
\bibitem [{\citenamefont {DiVincenzo}\ \emph {et~al.}(2002)\citenamefont {DiVincenzo}, \citenamefont {Leung},\ and\ \citenamefont {Terhal}}]{DLT02}%
  \BibitemOpen
  \bibfield  {author} {\bibinfo {author} {\bibfnamefont {D.}~\bibnamefont {DiVincenzo}}, \bibinfo {author} {\bibfnamefont {D.}~\bibnamefont {Leung}}, \ and\ \bibinfo {author} {\bibfnamefont {B.}~\bibnamefont {Terhal}},\ }\href {\doibase 10.1109/18.985948} {\bibfield  {journal} {\bibinfo  {journal} {IEEE Transactions on Information Theory}\ }\textbf {\bibinfo {volume} {48}},\ \bibinfo {pages} {580} (\bibinfo {year} {2002})}\BibitemShut {NoStop}%
\bibitem [{\citenamefont {Frank}\ and\ \citenamefont {Lieb}(2013)}]{FL13}%
  \BibitemOpen
  \bibfield  {author} {\bibinfo {author} {\bibfnamefont {R.~L.}\ \bibnamefont {Frank}}\ and\ \bibinfo {author} {\bibfnamefont {E.~H.}\ \bibnamefont {Lieb}},\ }\href {\doibase 10.1063/1.4838835} {\bibfield  {journal} {\bibinfo  {journal} {Journal of Mathematical Physics}\ }\textbf {\bibinfo {volume} {54}} (\bibinfo {year} {2013}),\ 10.1063/1.4838835}\BibitemShut {NoStop}%
\bibitem [{\citenamefont {Dupuis}\ \emph {et~al.}(2020)\citenamefont {Dupuis}, \citenamefont {Fawzi},\ and\ \citenamefont {Renner}}]{DFR20}%
  \BibitemOpen
  \bibfield  {author} {\bibinfo {author} {\bibfnamefont {F.}~\bibnamefont {Dupuis}}, \bibinfo {author} {\bibfnamefont {O.}~\bibnamefont {Fawzi}}, \ and\ \bibinfo {author} {\bibfnamefont {R.}~\bibnamefont {Renner}},\ }\href {\doibase 10.1007/s00220-020-03839-5} {\bibfield  {journal} {\bibinfo  {journal} {Communications in Mathematical Physics}\ }\textbf {\bibinfo {volume} {379}},\ \bibinfo {pages} {867} (\bibinfo {year} {2020})}\BibitemShut {NoStop}%
\bibitem [{\citenamefont {Pironio}\ \emph {et~al.}(2009)\citenamefont {Pironio}, \citenamefont {Ac\'in}, \citenamefont {Brunner}, \citenamefont {Gisin}, \citenamefont {Massar},\ and\ \citenamefont {Scarani}}]{PAB+09}%
  \BibitemOpen
  \bibfield  {author} {\bibinfo {author} {\bibfnamefont {S.}~\bibnamefont {Pironio}}, \bibinfo {author} {\bibfnamefont {A.}~\bibnamefont {Ac\'in}}, \bibinfo {author} {\bibfnamefont {N.}~\bibnamefont {Brunner}}, \bibinfo {author} {\bibfnamefont {N.}~\bibnamefont {Gisin}}, \bibinfo {author} {\bibfnamefont {S.}~\bibnamefont {Massar}}, \ and\ \bibinfo {author} {\bibfnamefont {V.}~\bibnamefont {Scarani}},\ }\href {\doibase 10.1088/1367-2630/11/4/045021} {\bibfield  {journal} {\bibinfo  {journal} {New Journal of Physics}\ }\textbf {\bibinfo {volume} {11}},\ \bibinfo {pages} {045021} (\bibinfo {year} {2009})}\BibitemShut {NoStop}%
\bibitem [{\citenamefont {Hahn}\ \emph {et~al.}(2025)\citenamefont {Hahn}, \citenamefont {Philip}, \citenamefont {Tan},\ and\ \citenamefont {Brown}}]{hahn_analytic_2025}%
  \BibitemOpen
  \bibfield  {author} {\bibinfo {author} {\bibfnamefont {T.~A.}\ \bibnamefont {Hahn}}, \bibinfo {author} {\bibfnamefont {A.}~\bibnamefont {Philip}}, \bibinfo {author} {\bibfnamefont {E.~Y.~Z.}\ \bibnamefont {Tan}}, \ and\ \bibinfo {author} {\bibfnamefont {P.}~\bibnamefont {Brown}},\ }\href {\doibase 10.48550/ARXIV.2507.07365} {\enquote {\bibinfo {title} {Analytic {R}ényi entropy bounds for device-independent cryptography},}\ } (\bibinfo {year} {2025})\BibitemShut {NoStop}%
\end{thebibliography}%

\cleardoublepage
\onecolumngrid

\appendix

\section{Proof of non-convexity}
\label{app:non-convex}

\begin{lemma}
\label{product_lemma}
Let $f(x)=ax+b$ with $a\neq 0$, and let $g:I\to\mathbb{R}$ (for some interval $I\subseteq\mathbb{R}$) be a twice-differentiable function such that there is an interior point $y_0$ of $I$ where $g'(y_0)\neq0$. Then
\begin{equation}
h(x,y)=f(x)g(y)
\end{equation}
is not jointly convex in $(x,y)$.
\end{lemma}

\begin{proof}
Since $f(x)=ax+b$, $f''(x)=0$. Thus, the Hessian of $h(x,y)$ is
\begin{equation}
\nabla^2 h(x,y)=
\begin{pmatrix}
0 & ag'(y) \\
ag'(y) & (ax+b)g''(y) \\
\end{pmatrix}.
\end{equation}

Its determinant is
\begin{equation}
\det(\nabla^2 h)
= -\bigl(ag'(y)\bigr)^2 \le 0.
\end{equation}

By hypothesis there exists an interior point $y_0$ such that $g'(y_0)\neq 0$. Thus at $(x,y_0)$, the determinant of the Hessian is strictly negative, implying that the two eigenvalues of the Hessian have opposite signs. 
Hence the Hessian is not positive semidefinite everywhere, and $h(x,y)$ cannot be jointly convex in $(x,y)$~\cite{boyd_convex_2004}.
\end{proof}

\begin{corollary}
\label{non-convex corollary}
The expression
\begin{align}
\label{eq:productterm}
q(\perp)\Bigg(
1 - \frac{1}{1-\beta}
\log\Big(
\left(1 - Q\right)^{\beta}
+ Q^{\beta}
\Big)
\Bigg)
\end{align}
is not jointly convex in $(q(\perp),Q)$.
\end{corollary}

\begin{proof}
It is obvious that $q(\perp)$ is a nonconstant linear function. Thus, we only need to prove that 
\begin{equation}
\label{eq:gterm}
1 - \frac{1}{1-\beta}
\log\Big(
\left(1 - Q\right)^{\beta}
+ Q^{\beta}
\Big)
\end{equation}
is a twice-differentiable function,  and that there exists $Q$ such that its first derivative is nonzero.

The first derivative was computed in~\eqref{eq:derivative}, and by inspection we see it is also differentiable, so \eqref{eq:gterm} is indeed twice differentiable. Moreover, clearly~\eqref{eq:derivative} is not identically zero over $Q\in(0,1)$.
Therefore, (\ref{eq:productterm}) satisfies the conditions of Lemma~\ref{product_lemma}, and is thus not jointly convex in $(q(\perp),Q)$.
\end{proof}

\section{Entropy bound for six-state protocol}
\label{app:sixstate}

There are multiple different versions of {\Renyi} entropy, with different applications. For this work, we discuss only the following versions, following the notation in~\cite{tomamichel_quantum_2016}:
\begin{equation}
\HPU_\alpha(B|E)_\rho = \frac{\alpha}{1-\alpha} \log Tr((Tr_B(\rho_{BE}^\alpha))^{\frac{1}{\alpha}})
\end{equation}

\begin{equation}
\tilde{H}^\downarrow _\alpha(B|E)_\rho = -\tilde{D}_\alpha ( \rho_{BE} \Vert \id_B \otimes \rho_E)
\end{equation}

\begin{equation}
\tilde{H}^\uparrow _\alpha(B|E)_\rho = - \inf_{\sigma_E} \tilde{D}_\alpha ( \rho_{BE} \Vert \id_B \otimes \sigma_E)
\end{equation}
where the infimum in the last line is taken over all quantum states $\sigma_E$, and 
\begin{equation}
\tilde{D}_\alpha ( \rho \Vert \sigma) = \frac{1}{\alpha-1} \log \frac{Tr\Bigg(\bigg[\sigma^{\frac{1-\alpha}{2\alpha}}\rho\sigma^{\frac{1-\alpha}{2\alpha}} \bigg]^\alpha 
\Bigg)}{Tr(\rho)}.
\end{equation}

The keyrate formula~\eqref{eq:keyrate6state} is given in terms of $\HSU_\alpha$. However, that {\Renyi} entropy can generally be challenging to evaluate, due to the infimum over $\sigma_E$. Fortunately, it is lower bounded by the other {\Renyi} entropies $\HPU_\alpha$ and $\HSD_\alpha$~\cite{tomamichel_largely_2017}, and thus we can analyse either of them in place of $\HSU_\alpha$ to obtain a lower bound. In our subsequent analysis, we consider both of these options, as well as an attempted approach to tackle $\HSU_\alpha$ directly.

\subsection{Reduction to Werner states}

It is commonly claimed that to analyse individual rounds of the six-state protocol, without loss of generality we can suppose that before Alice and Bob's measurements they initially share a Werner state, i.e.~a state of the following form for some $Q\in[0,2/3]$ (we discuss the physical interpretation of this parameter at the end of this section):
\begin{align}
\label{eq:werner}
(1-2Q) \ketbra{\Psi^-}{\Psi^-}_{AB} + 2Q \frac{\id_{AB}}{4},
\end{align}
and Eve holds a purification of it.
However, it is worth specifying in detail precisely how it applies in a security proof based on entropy accumulation, as such arguments can have subtle differences depending on the chosen proof framework.

Consider any initial state $\instate_{ABE}$ Alice and Bob could share in a single round. Consider the following attack Eve could perform to generate a different state $\widetilde{\instate}_{ABEG}$: Eve draws a uniformly random element of the Clifford group (on one qubit) and applies it to both $A$ and $B$, and stores her choice of element in a classical register $G$, i.e.
\begin{align}
\widetilde{\instate}_{ABEG} = \sum_{g \in \mathcal{C}_1} \frac{1}{\left|\mathcal{C}_1\right|} (g \otimes g \otimes \id) \instate_{ABE} (g^\dagger \otimes g^\dagger \otimes \id) \otimes \ketbra{g}{g}_G,
\end{align}
where the summation takes place over the one-qubit Clifford group $\mathcal{C}_1$. (This operation is often referred to as a ``Clifford twirl''.)
This state has the property that the reduced state $\widetilde{\instate}_{AB}$ is always a Werner state; see e.g.~\cite[Sec.~VI.A]{DLT02} for a proof. 

Furthermore, it has two critical properties in common with the original state $\instate_{AB}$. First, it has exactly the same probability of producing matched outcomes between Alice and Bob in test rounds of the six-state protocol we described: observe that this probability is given by the expectation of the operator\footnote{The factor of $1/3$ arises from choosing the Pauli basis uniformly at random in test rounds. It is important for our analysis here that the matching-outcome probability is ``averaged'' over the Pauli bases --- the Clifford twirl does not leave the \emph{individual} probabilities for each basis unchanged.}
\begin{align}\label{eq:matchPOVM}
\Gamma_\mathrm{match} = \sum_{k\in{X,Y,Z}} \frac{1}{3} \frac{\id_{AB} + \sigma_k \otimes (-\sigma_k)}{2},
\end{align}
and this expectation is left unchanged by the Clifford twirl (essentially, each Clifford gate simply ``permutes'' the Pauli operators in this sum). 
Second, suppose we perform a uniformly random choice $\ann$ of Pauli measurement on $A$ in\footnote{Here we shall trace out $B$ from the start because it is not involved in the relevant entropies. Also, throughout this discussion we shall suppose the $A$ register is always traced out after the measurement, as it is no longer involved anywhere in the analysis.} $\widetilde{\instate}_{AEG}$ and store the outcome in a classical register $S$, calling the resulting state $\widetilde{\outstate}_{S \ann EG|\texttt{gen}}$: here we used the $\texttt{gen}$ subscript to indicate that this is what would be produced in the protocol's generation rounds (we could equivalently have conditioned on $\CP=\perp$, but this notation is more compact). Crucially, this state is \emph{exactly equal to} $\Tr_{S'\ann'}[\outstate_{S'S \ann'\ann EG|\texttt{gen}}]$ where $\outstate_{S'S \ann'\ann EG|\texttt{gen}}$ is the state obtained by simply performing a $\sigma_{\ann'}$-measurement (for uniformly random $\ann'$) on $A$ in $\instate_{AE}$ and storing the outcome in a classical register $S'$, then drawing a uniformly random element $G$ of the Clifford group and setting $S=f_{\ann, G}(S')$ and $\ann = \tilde{f}_{G}(\ann')$ for suitable bijective functions $f_{\ann, G}$ and $\tilde{f}_{G}$.
(This is an immediate consequence of the fact that any Clifford operation simply permutes the Pauli group.)
With this, the following bound holds:
\begin{align}
\HSU_\alpha(S'|\ann' E)_{\outstate_{|\texttt{gen}}} 
&= \HSU_\alpha(S'|\ann' \ann EG)_{\outstate_{|\texttt{gen}}} \nonumber\\
&= \HSU_\alpha(SS'|\ann' \ann EG)_{\outstate_{|\texttt{gen}}} \nonumber\\
&\geq \HSU_\alpha(S|\ann' \ann EG)_{\outstate_{|\texttt{gen}}} \nonumber\\
&= \HSU_\alpha(S|\ann EG)_{\outstate_{|\texttt{gen}}} \nonumber\\
&= \HSU_\alpha(S|\ann EG)_{\widetilde{\outstate}_{|\texttt{gen}}},
\end{align}
where the first line holds by data-processing\footnote{The $\leq$ direction holds since $\ann$ is a function of $\ann' G$, and $G$ is independent of $S' \ann' E$.}~\cite{FL13}, the second line holds by~\cite[Lemma~B.7]{DFR20} (since $S$ was computed from $S'\ann G$), the third line holds by~\cite[Lemma~5.3]{tomamichel_quantum_2016} (since $S'$ is classical), 
the fourth line holds by data-processing~\cite{FL13} since $\ann'$ is completely determined by $\ann G$,
and the final line holds because $\widetilde{\outstate}_{|\texttt{gen}}$ and $\outstate_{|\texttt{gen}}$  are identical on $S\ann EG$ as mentioned above.

Summarizing these two critical properties in the context of the protocol:
\begin{enumerate}
\item If it is a test round, $\instate_{ABE}$ and $\widetilde{\instate}_{ABEG}$ produce the same outcome distribution on $\CP$.
\item If it is a generation round, the $\HSU_\alpha$ entropy of Alice's measurement outcome against Eve when measuring $\instate_{ABE}$ is lower bounded by that of measuring $\widetilde{\instate}_{ABEG}$ (with Eve holding $EG$ in the latter).
\end{enumerate}
Finally, we note that any purification $\ket{\widetilde{\instate}}_{ABEGR}$ of the state $\widetilde{\instate}_{ABEG}$ is also a valid purification of the Werner state $\widetilde{\instate}_{AB}$, and the result of measuring $A$ satisfies\footnote{Here we implicitly take ${\widetilde{\outstate}}_{SB\ann EGR|\texttt{gen}}$ to mean the state after measuring $A$ in $\ket{\widetilde{\instate}}_{ABEGR}$; it maintains consistency with the earlier definition of ${\widetilde{\outstate}_{|\texttt{gen}}}$ (which was only on $SB \ann EG$) because no operations acted on $R$.}
\begin{align}
\HSU_\alpha(S|\ann EG)_{\widetilde{\outstate}_{|\texttt{gen}}} \geq \HSU_\alpha(S|\ann EGR)_{\widetilde{\outstate}_{|\texttt{gen}}},
\end{align}
by data-processing~\cite{FL13}. 

With this, we conclude that for any initial state $\instate_{ABE}$ (before the measurements), the resulting entropy $\HSU_\alpha(S|\ann E)_{{\outstate}_{|\texttt{gen}}}$ (after the generation measurement) is lower bounded by the case  where the initial state on $AB$ is a Werner state $\widetilde{\instate}_{AB}$ and Eve holds a purification, where the value of $Q$ in that Werner state satisfies
\begin{equation}
\begin{aligned}
1-Q = \Tr[\Gamma_\mathrm{match}\instate_{AB}] = \frac{\outstate_{\CP}(0)}{\gamma} 
\quad\text{and}\quad
1-Q = \Tr[\Gamma_\mathrm{match}\widetilde{\instate}_{AB}] = \frac{\widetilde{\outstate}_{\CP}(0)}{\gamma},
\end{aligned}
\end{equation}
i.e.~$1-Q$ is equal to the original state's probability of yielding matching outcomes conditioned on being a test round, and also this probability for the corresponding Werner state. By normalization and the generation-round probability, we also have 
\begin{equation}
\begin{aligned}
Q = \frac{\outstate_{\CP}(1)}{\gamma} 
\quad\text{and}\quad
Q = \frac{\widetilde{\outstate}_{\CP}(1)}{\gamma}.
\end{aligned}
\end{equation}

\begin{remark}
\label{remark:werner}
There is a subtlety here: while the above equations might appear to suggest that $Q$ could lie in the entire interval $[0,1]$, recall that the Werner-state formula~\eqref{eq:werner} is only valid for $Q\in[0,2/3]$ (outside that range, it does not give a positive semidefinite operator). This reflects an important fact that there are \emph{no} quantum states $\instate_{AB}$ (even outside of Werner states) such that $\Tr\left[\Gamma_\mathrm{match} {\instate}_{AB}\right] < 1/3$.
\end{remark}

\subsection{{\Renyi} entropy from measuring a Werner state}
\label{app:WernerRenyi}

It thus suffices to find a bound on $\HSU_\alpha(S|\ann E)_\outstate$ under the restriction that the initial state $\instate_{ABE}$ before measurements is pure and the reduced state on $AB$ is a Werner state.
Since a Werner state is completely described by the single parameter $Q$ in~\eqref{eq:werner}, we shall find this bound in terms of $Q$.
To begin, we observe a Werner state is always a Bell-diagonal state, i.e.~a state of the form
\begin{align}
\sum_{j \in \{\Phi^+, \Phi^-, \Psi^+, \Psi^-\}} \lambda_{j} \ketbra{j}{j},
\end{align}
where
\begin{align} \label{eq:bellstates}
\ket{\Phi^\pm} \coloneqq \frac{\ket{00}\pm\ket{11}}{\sqrt{2}}, \quad 
\ket{\Psi^\pm} \coloneqq \frac{\ket{01}\pm\ket{10}}{\sqrt{2}}.
\end{align}
Comparing coefficients to~\eqref{eq:werner}, we see that for the Werner state, the $\lambda_j$ values satisfy
\begin{align}
\lambda_{\Psi^-} = 1-\frac{3Q}{2}, \quad \lambda_{\Phi^+} =\lambda_{\Phi^-} =\lambda_{\Psi^+} = \frac{Q}{2}.
\end{align}

We now wish to compute the entropies after a generation measurement. For this, we first note that it suffices to compute the result for a single choice of Pauli measurement, say, $\sigma_Z$, since by symmetry of the Werner state the resulting entropies are the same for each $\sigma_{\ann}$. We proceed by following the computations that were performed in~\cite{PAB+09}. Those computations were performed for the case where instead \emph{Bob} measures his qubit and we trace out Alice's qubit; for the sake of compatibility with those formulas, we shall retain that convention here --- note that this does not affect the entropies after the measurement, due to the symmetry between Alice and Bob in this scenario.
Moreover, for compactness of notation, we shall model this $\sigma_Z$ measurement via simply a pinching channel (in the $\sigma_Z$ eigenbasis) on $B$, i.e.~this allows us to view the measurement outcome as simply being directly encoded in the post-measurement state on $B$, without having to introduce a new classical register to store the outcome and then trace out $B$.
\newcommand{\classb}{b}

With these conventions in mind, the state on $E$ conditioned on outcome $\classb$ is~\cite{PAB+09}:
\begin{equation}
    \rho_{E \mid \classb} = \ket{\psi^{+}(\classb)}\bra{\psi^{+}(\classb)} + \ket{\psi^{-}(\classb)}\bra{\psi^{-}(\classb)}
\end{equation}
where
\begin{equation}
\begin{aligned}
\ket{\psi^{+}(1)} &= \sqrt{\frac{Q}{2}} \ket{e_1} + \sqrt{\frac{Q}{2}} \ket{e_2} \\
\ket{\psi^{-}(1)} &= \sqrt{\frac{Q}{2}} \ket{e_3} - \sqrt{1-\frac{3Q}{2}} \ket{e_4} \\
\ket{\psi^{+}(-1)} &= \sqrt{\frac{Q}{2}} \ket{e_3} + \sqrt{1-\frac{3Q}{2}} \ket{e_4} \\
\ket{\psi^{-}(-1)} &= \sqrt{\frac{Q}{2}} \ket{e_1} - \sqrt{\frac{Q}{2}} \ket{e_2} \\
\end{aligned}
\end{equation}

So the state on $BE$ is
\begin{equation}
\rho_{BE} = \sum_{\classb} p(\classb) \ket{\classb}\bra{\classb} \otimes \rho_{E \mid \classb}.
\end{equation}

For the first {\Renyi} entropy definition $\HPU_\alpha(B|E)_\rho$, 
we seek to find an expression for it as a function of Q.

We first find that

\begin{equation}
\rho_{BE}^\alpha = \sum_{\classb} p(\classb)^\alpha \ket{\classb}\bra{\classb} \otimes \rho_{E \mid \classb}^\alpha
\end{equation}

Given that $p(\classb) = \frac{1}{2}$ for $\classb = 1, -1$,

\begin{equation}
\label{eq:trace b}
\begin{aligned}
Tr_B(\rho_{BE}^\alpha) &= \sum_{\classb} p(\classb)^\alpha \rho_{E \mid \classb}^\alpha \\
&= 2^{-\alpha} \sum_{\classb} \rho_{E \mid \classb}^\alpha
\end{aligned}
\end{equation}

However to evaluate (\ref{eq:trace b}), we first find that

\begin{equation}
\label{eq:sum b}
\begin{aligned}
\sum_{\classb} \rho_{E \mid \classb}^\alpha &= Q^\alpha \ket{e_1} \bra{e_1} + Q^\alpha \ket{e_2} \bra{e_2}
+ Q(1-Q)^{\alpha-1} \ket{e_3} \bra{e_3} \\
&+ 2\bigg(1 - \frac{3Q}{2}\bigg)(1-Q)^{\alpha-1} \ket{e_4} \bra{e_4}
\end{aligned}
\end{equation}

By substituting (\ref{eq:sum b}) into (\ref{eq:trace b}), we can find that

\begin{equation}
\begin{aligned}
(Tr_B(\rho_{BE}^\alpha))^{\frac{1}{\alpha}} &= 2^{\frac{1}{\alpha}-1} \bigg[ 2^{-\frac{1}{\alpha}}Q \ket{e_1} \bra{e_1} 
+ 2^{-\frac{1}{\alpha}}Q \ket{e_2} \bra{e_2} 
+ \bigg( \frac{Q}{2} \bigg)^\frac{1}{\alpha} (1-Q)^{1-\frac{1}{\alpha}} \ket{e_3} \bra{e_3} \\
&+ \bigg( 1-\frac{3Q}{2} \bigg)^\frac{1}{\alpha} 
(1 - Q)^{1-\frac{1}{\alpha}} \ket{e_4} \bra{e_4} \bigg]
\end{aligned}
\end{equation}

\begin{equation}
\begin{aligned}
Tr((Tr_B(\rho_{BE}^\alpha))^{\frac{1}{\alpha}}) &= 2^{\frac{1}{\alpha}-1} \Bigg[2^{1-\frac{1}{\alpha}}Q
+ (1-Q)^{1-\frac{1}{\alpha}} \Bigg( \bigg( \frac{Q}{2} \bigg)^\frac{1}{\alpha}
+ \bigg( 1 - \frac{3Q}{2} \bigg)^\frac{1}{\alpha} \Bigg) \Bigg]
\end{aligned}
\end{equation}

Thus, an expression for $\HPU_\alpha(B | E)_\rho$ is found,

\begin{equation}
\label{eq:first renyi}
\begin{aligned}
\HPU_\alpha(B | E)_\rho &= \frac{\alpha}{1-\alpha}\log \Bigg[2^{1-\frac{1}{\alpha}}Q 
+ (1-Q)^{1-\frac{1}{\alpha}} \Bigg( \bigg( \frac{Q}{2} \bigg)^\frac{1}{\alpha}
+ \bigg( 1 - \frac{3Q}{2} \bigg)^\frac{1}{\alpha} \Bigg) \Bigg] + 1
\end{aligned}
\end{equation}

For the second {\Renyi} entropy definition $\tilde{H}^\downarrow _\alpha(B|E)_\rho$, similarly we seek to find an expression for it as a function of Q.

We first find that

\begin{equation}
\begin{aligned}
\rho_E &= Tr_B( \rho_{BE} ) \\
&= \frac{1}{2} \sum_{\classb} \rho_{E|\classb}
\end{aligned}
\end{equation}
and
\begin{equation}
\label{eq:product with e}
\id_B \otimes \rho_E = \frac{1}{2} \id_B \otimes \bigg( \sum_{\classb} \rho_{E|\classb} \bigg)
\end{equation}

It follows from (\ref{eq:product with e}) that

\begin{equation}
\label{eq:product power}
(\id_B \otimes \rho_E)^{\frac{1-\alpha}{2 \alpha}} = 
2^{-\frac{1-\alpha}{2\alpha}} \id_B \otimes \bigg( \sum_{\classb} \rho_{E|\classb} \bigg)^{\frac{1-\alpha}{2\alpha}}
\end{equation}

To evaluate (\ref{eq:product power}), we first find that

\begin{equation}
\label{eq: sum}
\begin{aligned}
\sum_{\classb} \rho_{E|\classb} &= Q \ket{e_1} \bra{e_1} + Q \ket{e_2} \bra{e_2}
+ Q \ket{e_3} \bra{e_3} \\
&+ 2\bigg(1 - \frac{3Q}{2}\bigg) \ket{e_4} \bra{e_4}
\end{aligned}
\end{equation}
and 
\begin{equation}
\label{eq:sum power}
\begin{aligned}
\bigg( \sum_{\classb} \rho_{E|\classb} \bigg)^{\frac{1-\alpha}{2\alpha}} &= Q^{\frac{1-\alpha}{2\alpha}} \ket{e_1} \bra{e_1} + Q^{\frac{1-\alpha}{2\alpha}} \ket{e_2} \bra{e_2} 
+ Q^{\frac{1-\alpha}{2\alpha}} \ket{e_3} \bra{e_3} \\
&+ \bigg[ 2\bigg(1 - \frac{3Q}{2}\bigg) \bigg]^{\frac{1-\alpha}{2\alpha}} \ket{e_4} \bra{e_4}
\end{aligned}
\end{equation}

By substituting (\ref{eq:sum power}) into (\ref{eq:product power}), we find that

\begin{equation}
\label{eq: shortcut for third}
\begin{aligned}
(\id_B \otimes \rho_E)^{\frac{1-\alpha}{2 \alpha}} \rho_{BE} (\id_B \otimes \rho_E)^{\frac{1-\alpha}{2 \alpha}}
= 2^{-\frac{1}{\alpha}} \sum_{\classb} \ket{\classb} \bra{\classb} \otimes 
\Bigg[ \bigg( \sum_{\classb} \rho_{E|\classb} \bigg)^{\frac{1-\alpha}{2\alpha}}
\rho_{E|\classb}
\bigg( \sum_{\classb} \rho_{E|\classb} \bigg)^{\frac{1-\alpha}{2\alpha}} \Bigg]
\end{aligned}
\end{equation}

\begin{equation}
\begin{aligned}
\Bigg[ (\id_B \otimes \rho_E)^{\frac{1-\alpha}{2 \alpha}} \rho_{BE} (\id_B \otimes \rho_E)^{\frac{1-\alpha}{2 \alpha}} \Bigg]^\alpha 
= 2^{-1} \sum_{\classb} \ket{\classb} \bra{\classb} \otimes 
\Bigg[ \bigg( \sum_{\classb} \rho_{E|\classb} \bigg)^{\frac{1-\alpha}{2\alpha}}
\rho_{E|\classb}
\bigg( \sum_{\classb} \rho_{E|\classb} \bigg)^{\frac{1-\alpha}{2\alpha}} \Bigg]^\alpha
\end{aligned}
\end{equation}

\begin{equation}
\begin{aligned}
Tr\Bigg( \Bigg[ (\id_B \otimes \rho_E)^{\frac{1-\alpha}{2 \alpha}} \rho_{BE} (\id_B \otimes \rho_E)^{\frac{1-\alpha}{2 \alpha}} \Bigg]^\alpha \Bigg) 
= Q + 2^{1-\alpha} \Bigg[ \bigg(\frac{Q}{2} \bigg)^{\frac{1}{\alpha}} + \bigg(1-\frac{3Q}{2} \bigg)^{\frac{1}{\alpha}} \Bigg]^\alpha
\end{aligned}
\end{equation}

Thus, an expression for $\tilde{D}_\alpha ( \rho_{BE} \Vert \id_B \otimes \rho_E)$ is found

\begin{equation}
\begin{aligned}
\tilde{D}_\alpha ( \rho_{BE} \Vert \id_B \otimes \rho_E) &= 
\frac{1}{\alpha-1} \log \Bigg( Q
+ 2^{1-\alpha} \Bigg[ \bigg(\frac{Q}{2} \bigg)^{\frac{1}{\alpha}}
+ \bigg(1-\frac{3Q}{2} \bigg)^{\frac{1}{\alpha}} \Bigg]^\alpha \Bigg)
\end{aligned}
\end{equation}
and the expression for $\tilde{H}^\downarrow _\alpha(B|E)_\rho$ follows
\begin{equation}
\label{eq:second renyi}
\begin{aligned}
\tilde{H}^\downarrow _\alpha(B|E)_\rho &= 
\frac{1}{1 - \alpha} \log \Bigg( Q + 2^{1-\alpha} \Bigg[ \bigg(\frac{Q}{2} \bigg)^{\frac{1}{\alpha}}
+ \bigg(1-\frac{3Q}{2} \bigg)^{\frac{1}{\alpha}} \Bigg]^\alpha \Bigg)
\end{aligned}
\end{equation}

For the third {\Renyi} entropy definition $\tilde{H}^\uparrow _\alpha(B|E)_\rho$, evaluating it explicitly can be challenging due to the inner optimization. Thus, we will evaluate a lower bound for it by using the following ansatz from~\cite{hahn_analytic_2025}:

\begin{equation}
\tilde{H}^\uparrow _\alpha(B|E)_\rho
\geq - \tilde{D}_\alpha ( \rho_{BE} \Vert \id_B \otimes \tau_E)
\end{equation}
where 
\begin{equation}
\label{eq:ansatz}
\tau_E = \frac{
\rho_E^{\frac{\alpha}{2 \alpha - 1}}
}{ Tr \bigg[ \rho_E^{\frac{\alpha}{2 \alpha - 1}} \bigg]}
\end{equation}

Then (leaving factors of $\mathbb{I}$ implicit for brevity):
\begin{equation}
\begin{aligned}
- \tilde{D}_\alpha ( \rho_{BE} \Vert \id_B \otimes \tau_E)
= \frac{1}{1 - \alpha}
\log \Bigg( \frac{
Tr \Bigg( \Bigg[ \rho_E^{\frac{1-\alpha'}{2\alpha'}} \rho_{BE} \rho_E^{\frac{1-\alpha'}{2\alpha'}} \Bigg]^\alpha \Bigg)
}{
Tr \bigg( \rho_E^{\frac{1}{\alpha'}} \bigg)^{1 - \alpha}
}
\Bigg)
\end{aligned}
\end{equation}

\begin{equation}
\alpha' = 2 - \frac{1}{\alpha}
\end{equation}

By using (\ref{eq: shortcut for third}) and replacing $\alpha \rightarrow\alpha'$, we can easily attain

\begin{equation}
\begin{aligned}
( \id_B \otimes \rho_E ) ^{\frac{1-\alpha'}{2\alpha'}} \rho_{BE} (\id_B \otimes \rho_E)^{\frac{1-\alpha'}{2\alpha'}}
= 2^{-\frac{1}{\alpha'}} \sum_{\classb} \ket{\classb} \bra{\classb} \otimes 
\Bigg[ \bigg( \sum_{\classb} \rho_{E|\classb} \bigg)^{\frac{1-\alpha'}{2\alpha'}}
\rho_{E|\classb}
\bigg( \sum_{\classb} \rho_{E|\classb} \bigg)^{\frac{1-\alpha'}{2\alpha'}} \Bigg]
\end{aligned}
\end{equation}

\begin{equation}
\label{eq:second tensor power}
\begin{aligned}
\Bigg[ ( \id_B \otimes \rho_E )^{\frac{1-\alpha'}{2\alpha'}} \rho_{BE} ( \id_B \otimes \rho_E )^{\frac{1-\alpha'}{2\alpha'}} \Bigg]^\alpha
= 2^{-\frac{\alpha}{\alpha'}} \sum_{\classb} \ket{\classb} \bra{\classb} \otimes 
\Bigg[ \bigg( \sum_{\classb} \rho_{E|\classb} \bigg)^{\frac{1-\alpha'}{2\alpha'}}
\rho_{E|\classb}
\bigg( \sum_{\classb} \rho_{E|\classb} \bigg)^{\frac{1-\alpha'}{2\alpha'}} \Bigg]^\alpha
\end{aligned}
\end{equation}

By substituting (\ref{eq:sum power}) and replacing $\alpha \rightarrow \alpha'$ into (\ref{eq:second tensor power}), 

\begin{equation}
\label{eq:numerator}
\begin{aligned}
Tr \Bigg(
\Bigg[ ( \id_B \otimes \rho_E )^{\frac{1-\alpha'}{2\alpha'}} \rho_{BE} ( \id_B \otimes \rho_E )^{\frac{1-\alpha'}{2\alpha'}} \Bigg]^\alpha
\Bigg) = 2^{-\frac{\alpha}{\alpha'}} \Bigg(
2Q^\frac{\alpha}{\alpha'} + 2^{\frac{(1-\alpha')\alpha}{\alpha'} + 1} \bigg[ 
\bigg(\frac{Q}{2} \bigg) ^ \frac{1}{\alpha'} + \bigg( 1 - \frac{3Q}{2} \bigg) ^ \frac{1}{\alpha'}
\bigg] ^ \alpha
\Bigg)
\end{aligned}
\end{equation}

Next, we will evaluate

\begin{equation}
\begin{aligned}
\rho_E^\frac{1}{\alpha'} &= 2^{-\frac{1}{\alpha'}} \bigg( \sum_{\classb} \rho_{E|\classb} \bigg) ^ \frac{1}{\alpha'}
\end{aligned}
\end{equation}

By substituting (\ref{eq: sum}),

\begin{equation}
\begin{aligned}
\bigg( \sum_{\classb} \rho_{E|\classb} \bigg) ^ \frac{1}{\alpha'} &= Q^\frac{1}{\alpha'} \ket{e_1} \bra{e_1} + Q^\frac{1}{\alpha'} \ket{e_2} \bra{e_2}
+ Q^\frac{1}{\alpha'} \ket{e_3} \bra{e_3} \\
&+ 2^\frac{1}{\alpha'}
\bigg(1 - \frac{3Q}{2}\bigg)^\frac{1}{\alpha'} \ket{e_4} \bra{e_4}
\end{aligned}
\end{equation}

\begin{equation}
\begin{aligned}
Tr \bigg( \rho_E^{\frac{1}{\alpha'}} \bigg) = 2^{-\frac{1}{\alpha'}} \bigg[ 
3Q^\frac{1}{\alpha'} + 2^\frac{1}{\alpha'}\bigg( 1 - \frac{3Q}{2} \bigg)^\frac{1}{\alpha'}
\bigg]
\end{aligned}
\end{equation}

\begin{equation}
\label{eq:denominator}
\begin{aligned}
Tr \bigg( \rho_E^{\frac{1}{\alpha'}} \bigg) ^{1-\alpha}
= 2^{-\frac{1-\alpha}{\alpha'}} \bigg[ 
3Q^\frac{1}{\alpha'} + 2^\frac{1}{\alpha'}\bigg( 1 - \frac{3Q}{2} \bigg)^\frac{1}{\alpha'}
\bigg]^{1-\alpha}
\end{aligned}
\end{equation}

By combining (\ref{eq:numerator}) and (\ref{eq:denominator}),

\begin{equation}
\label{eq:third renyi}
\begin{aligned}
- \tilde{D}_\alpha ( \rho_{BE} \Vert \id_B \otimes \tau_E)
&= \frac{1}{1-\alpha} \log \Bigg(
2Q^\frac{\alpha}{\alpha'} + 2^{\frac{(1-\alpha')\alpha}{\alpha'} + 1} \bigg[ 
\bigg(\frac{Q}{2} \bigg) ^ \frac{1}{\alpha'} + \bigg( 1 - \frac{3Q}{2} \bigg) ^ \frac{1}{\alpha'}
\bigg] ^ \alpha
\Bigg) \\
&- \log \bigg( 3Q^\frac{1}{\alpha'} + 2^\frac{1}{\alpha'}\bigg( 1 - \frac{3Q}{2} \bigg)^\frac{1}{\alpha'} 
\bigg) - \frac{\alpha}{1-\alpha}
\end{aligned}
\end{equation}

Thus a lower bound for $\tilde{H}^\uparrow _\alpha(B|E)_\rho$ has been found.

Finally, we note that for our numerical methods, it is important for the bounds to be convex in $Q$. This is indeed the case for $\HPU_\alpha(B|E)$ and $\HSD_\alpha(B|E)$: it follows somewhat abstractly from the observation that~\cite[Lemma~16]{kamin_renyi_2025} implies that both these quantities in this protocol are equal to convex functions of $\instate_{AB}$. Since a Werner state $\instate_{AB}$ is an affine function of $Q$, this in turn implies that these entropies are both convex functions of $Q$.

On the other hand, it is less clear whether the bound we obtained from $\HSU_\alpha(B|E)$ is convex. This is because we did not compute its exact value (which would have allowed us to apply~\cite[Lemma~16]{kamin_renyi_2025}), but rather a lower bound on it via the ansatz $\tau_E$ in~\eqref{eq:ansatz}. We leave the convexity of that bound for future work, as our empirical observations below indicate it performs less well than the other bounds in practice.

\subsection{Comparison of bounds}
\label{subsec:boundperformance}

For our numerical computations, we in fact computed the keyrates using all three of these bounds, i.e.~(\ref{eq:first renyi}), (\ref{eq:second renyi}) and (\ref{eq:third renyi}). Empirically, we found that the first bound performed best. Therefore, in the main text we chose to express this by setting the bound $\sixstate{\alpha}(Q)$ to that expression,
\begin{align}
\label{eq:sixstatebnd}
\sixstate{\alpha}(Q) = \frac{\alpha}{1-\alpha}\log \Bigg[2^{1-\frac{1}{\alpha}}Q 
+ (1-Q)^{1-\frac{1}{\alpha}} \Bigg( \bigg( \frac{Q}{2} \bigg)^\frac{1}{\alpha}
+ \bigg( 1 - \frac{3Q}{2} \bigg)^\frac{1}{\alpha} \Bigg) \Bigg] + 1.
\end{align}

We again emphasise the point mentioned in the main text that any inequalities between the three bounds~(\ref{eq:first renyi}), (\ref{eq:second renyi}) and (\ref{eq:third renyi}) are not necessarily preserved when approximating them via a finite set of tangent lines. In fact, for $n=10^7$, (\ref{eq:second renyi}) was found to give a higher keyrate than~(\ref{eq:first renyi}), despite the 
proof we present below
that the expression (\ref{eq:second renyi}) is smaller than (\ref{eq:first renyi}). This demonstrates that, when only a finite number of affine lower bounds are considered, the lower bound obtained for a higher function need not be greater than that obtained for a smaller function.

Still, despite the above consideration, we believe it may be of interest to have a proof of some relations between the bounds. Therefore, we now present a proof that (\ref{eq:second renyi}) is smaller than (\ref{eq:first renyi}). We leave further study of the third bound (\ref{eq:third renyi}) to future work.
\begin{lemma}
\label{inequality lemma}
The following bound holds:
\begin{equation}
\begin{aligned}
&\frac{\alpha}{1-\alpha}\log \Bigg[2^{1-\frac{1}{\alpha}}Q 
+ (1-Q)^{1-\frac{1}{\alpha}} \Bigg( \bigg( \frac{Q}{2} \bigg)^\frac{1}{\alpha}
+ \bigg( 1 - \frac{3Q}{2} \bigg)^\frac{1}{\alpha} \Bigg) \Bigg] + 1 \\
&\geq
\frac{1}{1 - \alpha} \log \Bigg( Q + 2^{1-\alpha} \Bigg[ \bigg(\frac{Q}{2} \bigg)^{\frac{1}{\alpha}}
+ \bigg(1-\frac{3Q}{2} \bigg)^{\frac{1}{\alpha}} \Bigg]^\alpha \Bigg)
\end{aligned}
\end{equation}
given that $\alpha > 1$ and $0 \leq Q \leq \frac{2}{3}$.
\end{lemma}

\begin{proof}

In this proof, we shall use H\"{o}lder's inequality in the following form: for any two sequences of numbers $a_i$ and $b_i$, we have
\begin{align}
\sum^{n}_{i=1} |a_i b_i| \leq \bigg(\sum^{n}_{i=1} |a_i|^p \bigg)^\frac{1}{p} \bigg(\sum^{n}_{i=1} |b_i|^q \bigg)^\frac{1}{q}
\end{align}
where $p$ and $q$ are real numbers satisfying, $p,q\geq1$ and $\frac{1}{p} + \frac{1}{q} = 1$.

We will first manipulate the expression into another form. Since $\alpha > 1$, $1 - \alpha < 0$, we rewrite the desired inequality as follows:
\begin{equation}
\begin{aligned}
&\alpha \log \Bigg[2^{1-\frac{1}{\alpha}}Q 
+ (1-Q)^{1-\frac{1}{\alpha}} \Bigg( \bigg( \frac{Q}{2} \bigg)^\frac{1}{\alpha}
+ \bigg( 1 - \frac{3Q}{2} \bigg)^\frac{1}{\alpha} \Bigg) \Bigg] + 1 - \alpha \\
&\leq
\log \Bigg( Q + 2^{1-\alpha} \Bigg[ \bigg(\frac{Q}{2} \bigg)^{\frac{1}{\alpha}}
+ \bigg(1-\frac{3Q}{2} \bigg)^{\frac{1}{\alpha}} \Bigg]^\alpha \Bigg)
\end{aligned}
\end{equation}

\begin{equation}
\begin{aligned}
\iff 2^{1 - \alpha}\Bigg[2^{1-\frac{1}{\alpha}}Q 
+ (1-Q)^{1-\frac{1}{\alpha}} \Bigg( \bigg( \frac{Q}{2} \bigg)^\frac{1}{\alpha}
+ \bigg( 1 - \frac{3Q}{2} \bigg)^\frac{1}{\alpha} \Bigg) \Bigg]^\alpha
\leq
Q + 2^{1-\alpha} \Bigg[ \bigg(\frac{Q}{2} \bigg)^{\frac{1}{\alpha}}
+ \bigg(1-\frac{3Q}{2} \bigg)^{\frac{1}{\alpha}} \Bigg]^\alpha
\end{aligned}
\end{equation}

\begin{equation}
\begin{aligned}
\iff  \Bigg[2^{1-\frac{1}{\alpha}}Q 
+ (1-Q)^{1-\frac{1}{\alpha}} \Bigg( \bigg( \frac{Q}{2} \bigg)^\frac{1}{\alpha}
+ \bigg( 1 - \frac{3Q}{2} \bigg)^\frac{1}{\alpha} \Bigg) \Bigg]^\alpha
\leq
2^{\alpha - 1}Q + \Bigg[ \bigg(\frac{Q}{2} \bigg)^{\frac{1}{\alpha}}
+ \bigg(1-\frac{3Q}{2} \bigg)^{\frac{1}{\alpha}} \Bigg]^\alpha
\end{aligned}
\end{equation}

\begin{equation}
\label{eq:intermediate ineq}
\begin{aligned}
\iff 2^{1-\frac{1}{\alpha}}Q 
+ (1-Q)^{1-\frac{1}{\alpha}} \Bigg( \bigg( \frac{Q}{2} \bigg)^\frac{1}{\alpha}
+ \bigg( 1 - \frac{3Q}{2} \bigg)^\frac{1}{\alpha} \Bigg)
\leq
\Bigg[ 2^{\alpha - 1}Q + \Bigg[ \bigg(\frac{Q}{2} \bigg)^{\frac{1}{\alpha}}
+ \bigg(1-\frac{3Q}{2} \bigg)^{\frac{1}{\alpha}} \Bigg]^\alpha \Bigg]^\frac{1}{\alpha}
\end{aligned}
\end{equation}

We will substitute 
\begin{equation}
    A = 2^{1-\frac{1}{\alpha}}
\end{equation}
\begin{equation}
    B = \bigg(\frac{Q}{2} \bigg)^{\frac{1}{\alpha}}
+ \bigg(1-\frac{3Q}{2} \bigg)^{\frac{1}{\alpha}}
\end{equation}

Thus, \eqref{eq:intermediate ineq} becomes

\begin{equation}
\label{eq:holder ineq}
\begin{aligned}
AQ + B(1-Q)^{1-\frac{1}{\alpha}} \leq \bigg[ A^\alpha Q + B^\alpha \bigg]^\frac{1}{\alpha}
\end{aligned}
\end{equation}

Next, we can use Hölder's Inequality to prove \eqref{eq:holder ineq}. Since \(Q\), \(1-Q\), \(A\), and \(B\) are all non-negative for \(\alpha>1\) and \(0\leq Q\leq \frac{2}{3}\), the conditions for Hölder's inequality are satisfied.

\begin{equation}
\begin{aligned}
\bigg[ A^\alpha Q + B^\alpha \bigg]^\frac{1}{\alpha}
= \bigg[ A^\alpha Q + B^\alpha \bigg]^\frac{1}{\alpha}
[Q + 1 - Q]^{1 - \frac{1}{\alpha}}
\end{aligned}
\end{equation}

\begin{equation}
\begin{aligned}
 \bigg[ A^\alpha Q + B^\alpha \bigg]^\frac{1}{\alpha} [Q + 1 - Q]^{1 - \frac{1}{\alpha}}
 \geq A Q^\frac{1}{\alpha} Q^{1-\frac{1}{\alpha}} + B(1-Q)^{1-\frac{1}{\alpha}}
[Q + 1 - Q]^{1 - \frac{1}{\alpha}}
\end{aligned}
\end{equation}

Therefore,
\begin{equation}
\begin{aligned}
 \bigg[ A^\alpha Q + B^\alpha \bigg]^\frac{1}{\alpha}
 \geq A Q+ B(1-Q)^{1-\frac{1}{\alpha}}
\end{aligned}
\end{equation}
and the desired inequality holds.
\end{proof}

\end{document}